\documentclass[compress,sort,twocolumn]{autart}

\usepackage{graphicx}          
\usepackage{subfigure}
\usepackage{fancyhdr}
\usepackage{amsmath}
\usepackage{multirow}
\usepackage{amssymb }
\usepackage{color}
\usepackage{graphics} 
\usepackage{graphicx}
\usepackage{amsmath}
\newtheorem{theorem}{\textbf{Theorem}}
\newtheorem{lemma}{\textbf{Lemma}}
\newtheorem{example}{\textbf{Example}}
\newtheorem{corollary}{\textbf{Corollary}}
\newtheorem{remark}{\textbf{Remark}}
\newtheorem{definition}{\textbf{Definition}}
\newtheorem{problem}{\textbf{Problem}}
\newtheorem{proposition}{\textbf{Proposition}}
\newtheorem{assumption}{\textbf{Assumption}}
\usepackage{url}
	\newenvironment{proof}{{{\bf Proof:}}}{\hfill $\square$\par}
	
	\usepackage{multirow} 
	\usepackage{xcolor}
	\usepackage{multirow}
	\usepackage{setspace}
	\usepackage{url}
	\usepackage{subfigure}
	\usepackage{fancyhdr}
	\usepackage{amsmath}
	\usepackage{multirow}
	\usepackage{amssymb }
	\usepackage{color}
	\usepackage{graphics} 
	\usepackage{graphicx}
	\usepackage{subeqnarray}
	\usepackage{cases}
	\usepackage{threeparttable}
	\usepackage{algorithm}
	\usepackage{algpseudocode}

\usepackage{graphicx}
\usepackage{needspace}

\newsavebox{\authorphotobox}
\newlength{\authorphototop}
\newcount\authorphotolines

\providecommand{\authorbiography}[3]{}
\renewcommand{\authorbiography}[3]{%
  \par\addvspace{3pt}%
  \Needspace{36mm}%
  \begingroup
    \setlength{\parindent}{0pt}%
    \setlength{\parskip}{0pt}%
    \sbox{\authorphotobox}{%
      \parbox[c][33mm][c]{25mm}{%
        \centering
        \IfFileExists{#1}{%
          \includegraphics[
            width=25mm,
            height=33mm,
            keepaspectratio
          ]{#1}%
        }{%
          \begingroup
          \setlength{\fboxsep}{0pt}%
          \fbox{%
            \parbox[c][\dimexpr33mm-2\fboxrule\relax][c]
            {\dimexpr25mm-2\fboxrule\relax}{%
              \centering\footnotesize Photo
            }%
          }%
          \endgroup
        }%
      }%
    }%
    \authorphotolines=\dimexpr33mm+\baselineskip\relax
    \divide\authorphotolines by \baselineskip
    \settoheight{\authorphototop}{T}%
    \noindent
    \parshape=0\relax
    \hangindent=28mm\relax
    \hangafter=-\authorphotolines\relax
    \makebox[0pt][l]{%
      \hspace*{-28mm}%
      \raisebox
        {\dimexpr\authorphototop-\ht\authorphotobox\relax}
        [0pt][0pt]{\usebox{\authorphotobox}}%
    }%
    \textbf{#2} #3\par
    \ifnum\prevgraf<\authorphotolines
      \advance\authorphotolines by -\prevgraf
      \vspace*{\authorphotolines\baselineskip}%
    \fi
    \par
  \endgroup
}

	 \normalsize
	\begin{document}
		
		\begin{frontmatter}


      \title{{\large{Weak and Strong Data Informativity with Structural Priors \\~ for System Analysis}}}
			
			
			\thanks[footnoteinfo]{This work was supported in part by the National Key Research and Development Program of China under Grant 2025YFA1018800, the National
				Natural Science Foundation of China under Grant 62373059, and the Beijing Institute of Technology Research Fund Program for Young Scholars.}
			
			{\small{
					{\author[bit]{ Yuan Zhang}\ead{zhangyuan14@bit.edu.cn},    \author[bit]{Wenxuan Xu}\ead{xuwenxuan1209@163.com}, \author[hunan]{Yun Feng}\ead{fyrobot@hnu.edu.cn},
					\author[qingdao]{Keke Huang}\ead{huangkeke@csu.edu.cn},  
					\author[thu]{Keyou You}\ead{
youky@tsinghua.edu.cn}}
			}} 
			
			\address[bit]{School of Automation, Beijing Institute of Technology, Beijing 100081, China}
			\address[hunan]{School of Artificial Intelligence and Robotics, and the National Engineering Research Center for Robot Visual Perception and Control Technology, Hunan University, Changsha 410082, China}
			\address[qingdao]{School of Automation, Central South University, Changsha 410083, China.} 
			\address[thu]{Department of Automation, and Beijing National Research Center for Information Science and Technology, Tsinghua University, Beijing 100084, China}

			\begin{abstract} 
When measurement data are insufficient for unique system identification, data informativity evaluates whether a dataset determines a property across all systems consistent with the data. This paper introduces a unified framework distinguishing between weak and strong data informativity under prior structural knowledge described by affine constraints: data are weakly informative if at least one system consistent with the data and prior knowledge possesses the target property, and strongly informative if every consistent system satisfies it. The latter is consistent with the notion of data informativity in the existing literature. By explicitly parameterizing all consistent systems as affine functions of free parameters via the Moore-Penrose pseudoinverse, we characterize these notions for fundamental system properties—including controllability, stabilizability, and stability. This formulation establishes that verifying weak and strong informativity, when tractable, reduces to rank tests of matrix pencils constructed from the data and prior knowledge. Our findings show that structural priors can lower data-excitation requirements and enable reliable analysis even when data fail to be strongly informative. Finally, the framework is extended to a class of noisy data.

			\end{abstract}
		\end{frontmatter}
		{{

\section{Introduction}

For decades, model-based control has dominated control engineering, relying heavily on accurate mathematical models derived from first principles or system identification. While an established model enables rigorous controller design and analysis using mature tools, developing such models can be time-consuming, resource-intensive, and prone to modeling inaccuracies when high-fidelity physics are unattainable \cite{6093751,8358007}. In response to these challenges, data-driven control has emerged as an attractive alternative paradigm \cite{2013Model}. Rather than requiring an explicit mathematical model, data-driven methods seek to analyze systems and design controllers directly from observational data collected during operation, thereby bypassing the potentially difficult intermediate step of first-principles modeling \cite{formulas2020,kang2023minimum}.

Generically, data-driven methods can be classified into direct and indirect approaches \cite{krishnan2021direct}. Direct methods evaluate system properties or design controllers directly from data without explicitly identifying a unique system model, whereas indirect methods follow the traditional identify-then-control pipeline. When data are sufficiently informative or persistently exciting---such that the dataset uniquely identifies the underlying system---both paradigms yield equivalent outcomes. However, when unique identification is precluded due to short experiments, insufficient excitation, or operational constraints, the seminal work of van Waarde et al.~\cite{informativity2020} introduced the framework of data informativity. This framework characterizes whether a given dataset is sufficient to determine a specific system property or controller across all systems consistent with the data. Although subsequently extended to diverse system properties and control tasks \cite{van2022data,zhu2024data}, existing informativity frameworks predominantly adopt a black-box or data-centric viewpoint, overlooking a vital resource in engineering practice: prior structural knowledge (e.g., physical laws, structural constraints, or known structural interconnections). Consequently, developing systematic methodologies that jointly exploit observational data and structural priors remains a critical and open challenge.

{A closer inspection of informativity for structural system properties such as controllability and stability reveals a deeper limitation. These conditions are typically formulated as rank conditions on data matrices \cite{informativity2020}. When these conditions fail because of short experiments or insufficient excitation, the data cannot certify the property in the standard sense. Crucially, such a failure is entirely ambiguous: it does not distinguish whether every data-consistent system violates the property or at least one consistent system possesses it. As a result, datasets conventionally deemed ``non-informative'' are often discarded or assumed useless, overlooking potentially valuable information and prompting costly additional experiments.

This observation motivates the central research question of this paper: What definitive conclusions can be drawn about system properties from data that fail to be informative in the standard sense?

To answer this question, we depart from the binary viewpoint of existing informativity and introduce a unified framework distinguishing between weak and strong data informativity in the presence of prior structural knowledge. Specifically, strong informativity requires every system consistent with the data and priors to possess the target property (consistent with the existing data informativity \cite{informativity2020}). In contrast, weak informativity requires only existential consistency---that is, at least one data- and prior-consistent system possesses the property. For fundamental properties such as controllability, stabilizability, and stability, weak informativity rigorously delineates the exact algebraic boundary between guaranteed property satisfaction and its guaranteed violation,  closing the analytical gap left by standard informativity. The main contributions are summarized as follows.

\begin{enumerate}
\item \textbf{Introducing weak data informativity.}
We introduce weak data informativity and derive necessary and sufficient conditions for controllability, stabilizability, and stability. For controllability, we establish a generic dichotomy: either no data-consistent system is controllable, or almost all such systems are controllable. Weak informativity therefore characterizes the intermediate regime between strong informativity for a property and strong informativity for its complement.

\item \textbf{Incorporating prior structural knowledge.}
We represent structural prior knowledge by affine constraints on the system matrices and parameterize the systems jointly consistent with the data and the prior. This yields characterizations of weak and strong informativity and reveals the computational limitations of general affine priors. For the tractable cases associated with single-affine constraints, we derive necessary and/or sufficient conditions based on matrix-pencil rank tests.

\item \textbf{Extending the framework to noisy data.}
We extend the framework to data corrupted by noise lying in a known linear subspace. In the absence of structural priors, weak informativity can still provide conclusions about controllability and stability in this setting, whereas strong informativity for these properties is possible only when the data are noiseless.
\end{enumerate}

The central technical innovation of this paper lies in explicitly parameterizing all systems jointly consistent with observational data and linear structural priors as affine functions of free parameters via the Moore--Penrose pseudoinverse. This characterization maps the space of consistent systems onto a structured parameter manifold, thereby transforming abstract property-verification problems into (tractable) algebraic conditions. Overall, our results show that meaningful conclusions about system properties can still be drawn from conventionally non-informative data and that structural prior knowledge can reduce the data excitation requirements for system analysis.}

{
{A growing body of work combines prior knowledge with data in systems and control \cite{berberich2022combining,chiuso2025harnessing,Breschi2023StructuralPriors,Kogan2025Robust,miller2025data,shakouri2026data,huang2025datadriven,vandenhof2026informativity,cao2026informativity}.} Existing studies exploit constraints on system matrices or uncertainties \cite{berberich2022combining,huang2025datadriven}, controller structures \cite{Breschi2023StructuralPriors,miller2025data}, stochastic distributions \cite{chiuso2025harnessing}, or controllability and stabilizability \cite{shakouri2026data}, mainly for controller design. More closely related, \cite{vandenhof2026informativity} uses known network structure to relax informativity conditions for target-module identification, whereas \cite{cao2026informativity} studies conventional informativity and controller synthesis for systems with partially known matrix entries. In contrast, this paper considers general affine structural priors, derives tractable conditions for the typical class of single-affine constraints, and characterizes both weak and strong data informativity for controllability and stability.
}

The paper is organized as follows. Section II presents the problem formulation. Section III characterizes weak data informativity for controllability, stabilizability, and stability. Section IV investigates weak data informativity in the presence of prior structural knowledge, while Section V studies strong data informativity under the same prior knowledge. Section VI extends the proposed framework to a class of noisy data, and Section VII concludes the paper.

{\bf Notations.}
For a matrix $M$, let $M^\dagger$ be the Moore--Penrose pseudoinverse, ${\rm im}(M)$ the column space of $M$, and ${\rm Tr}(M)$ the trace. Let $M^T$ denote the transpose of matrix $M$ and $\mathrm{rank}(M)$ the rank. Furthermore, we denote by $M^\perp=I-MM^\dagger$ and $M_\perp=I-M^\dagger M$. For a linear map $\Phi$, ${\rm im}(\Phi)$ denotes its image and ${\rm ker}\,\Phi$ denotes its kernel. For a linear space $\mathcal S$, ${\rm dim}(\mathcal S)$ denotes its dimension.

\section{Problem formulation}\label{sec-prob}
\begingroup
This section describes how observational data and structural priors restrict the admissible systems, and introduces weak and strong data informativity on the resulting set.

\subsection{System dynamics and observational data}

Consider the discrete-time linear time-invariant system
\begin{equation}\label{eq:sys}
    x(t+1)=A_u x(t)+B_u u(t),
\end{equation}
where $x(t)\in\mathbb{R}^n$ and $u(t)\in\mathbb{R}^m$ denote the state and input, respectively, and the true matrices $(A_u,B_u)$ are unknown. Given an input--state trajectory $\mathcal D=\{U,X\}$, with $U=\{u(0),\ldots,u(T-1)\}$ and $X=\{x(0),\ldots,x(T)\}$, define
\begin{align*}
X_p&:=\begin{bmatrix}x(0)&\cdots&x(T-1)\end{bmatrix},\\
X_f&:=\begin{bmatrix}x(1)&\cdots&x(T)\end{bmatrix},\\
U_p&:=\begin{bmatrix}u(0)&\cdots&u(T-1)\end{bmatrix},
\end{align*}
and $Z_p:=[X_p^T,U_p^T]^T\in\mathbb R^{(n+m)\times T}$. In the absence of noise, these matrices satisfy 
\begin{equation}\label{eq:data-equation}
    X_f=A_uX_p+B_uU_p=[A_u,B_u]Z_p.
\end{equation}
Let $\Sigma:=\mathbb R^{n\times n}\times\mathbb R^{n\times m}$ denote the ambient system space. The systems compatible with $\mathcal D$ form the \emph{data-consistent set}
\begin{equation}\label{eq:data-consistent-set}
    \Sigma_{\mathcal D}:=\{(A,B)\in\Sigma\mid [A,B]Z_p=X_f\}.
\end{equation}
If $\operatorname{rank}(Z_p)=n+m$, the data uniquely determine $[A_u,B_u]=X_fZ_p^\dagger$. Otherwise, $\Sigma_{\mathcal D}$ is an affine set containing infinitely many systems. All subsequent results apply to multiple trajectories by concatenating their respective data matrices horizontally. Note that each additional trajectory can potentially shrink $\Sigma_{\mathcal D}$.

\subsection{Prior structural knowledge}

Physical laws and network architecture often specify entries or relations among entries of $(A,B)$ even when its numerical parameters are unknown~\cite{berberich2022combining,zhang2019structural}. We represent such knowledge by affine equality constraints.

\begin{assumption}[General affine constraint]\label{general_prior}
There are known matrices $M_i\in\mathbb R^{n\times(n+m)}$ and a vector $\gamma=(\gamma_1,\ldots,\gamma_r)^T\in\mathbb R^r$ such that every admissible system satisfies
\begin{equation}\label{eq:structural_constraint}
    \langle M_i,[A,B]\rangle=\gamma_i,\qquad i=1,\ldots,r,
\end{equation}
where $\langle P,Q\rangle=\operatorname{Tr}(P^TQ)$ is the Frobenius inner product.
Denote the corresponding set by
\[
\Sigma_{\rm stru}:=\{(A,B)\in\Sigma\mid\langle M_i,[A,B]\rangle=\gamma_i,\ i=1,\ldots,r\}.
\]
\end{assumption}

For example, known entries, including zeros imposed by network topology, are obtained by choosing $M_i$ with a single nonzero entry. Symmetry and prescribed row sums are also affine constraints. Equivalently, the admissible matrices can be written as $[A,B]=N_0+\sum_{i=1}^{l}\mu_iN_i$, with known matrices $N_i$ and free scalars $\mu_i$~\cite{zhang2023real,miller2025data}.

We also consider a subclass for which more explicit criteria can be derived.
\begin{assumption}[Single-affine prior constraint]\label{assump:structured_block}
There are known matrices $S_\ell\in\mathbb R^{n_l\times n}$, $S_r\in\mathbb R^{(n+m)\times n_r}$, and $Q\in\mathbb R^{n_l\times n_r}$ such that
\begin{equation}\label{eq:Block_structural_constraint}
    S_\ell[A\,\,B]S_r=Q.
\end{equation}
The corresponding prior set is
\[
\hat\Sigma_{\rm stru}:=\{(A,B)\in\Sigma\mid S_\ell[A\,\,B]S_r=Q\}.
\]
\end{assumption}

Row and column selectors $S_\ell,S_r$ fix a prescribed subblock of $[A,B]$.  Setting $S_{\ell}=I_n$ (or $S_r=I_{n+m}$) imposes linear constraints on the columns (or rows) of $[A,B]$. For example, taking $S_\ell=I_n$, $S_r=[\mathbf 1_n^T,0_m^T]^T$, and $Q=\mathbf 1_n$ imposes $A\mathbf 1_n=\mathbf 1_n$, the row-sum constraint of a row-stochastic matrix, commonly arising in
discrete-time Laplacian dynamics.

Write $\Sigma_{\mathrm{prio}}$ for the prescribed prior set, such as $\Sigma_{\rm stru}$ or $\hat\Sigma_{\rm stru}$, and take $\Sigma_{\mathrm{prio}}=\Sigma$ when no prior is available. The jointly consistent systems form $\Sigma_{\mathcal D}\cap\Sigma_{\mathrm{prio}}$. Throughout, the prior is assumed valid for $(A_u,B_u)$, so this intersection is nonempty.

\subsection{Weak and strong data informativity}

For a property $\mathcal P$, let $\Sigma_{\mathcal P}\subseteq\Sigma$ denote the set of systems possessing it. We consider controllability, stability, and stabilizability, with
\begin{align*}
\Sigma_{\rm cont}&:=\{(A,B)\in\Sigma\mid(A,B)\text{ is controllable}\},\\
\Sigma_{\rm stab}&:=\{(A,B)\in\Sigma\mid\rho(A)<1\},\\
\Sigma_{\rm stbz}&:=\{(A,B)\in\Sigma\mid\exists K\in\mathbb R^{m\times n},\ \rho(A+BK)<1\},
\end{align*}
where $\rho(\cdot)$ denotes the spectral radius.

\begin{definition}[Weak and strong informativity]\label{def:weak-informativity}\label{def:weak-informativity-prior}\label{def:strong-informativity-prior}
The data $(U,X)$ are \emph{strongly informative for $\mathcal P$ on $\Sigma_{\mathrm{prio}}$} if every jointly consistent system possesses $\mathcal P$, i.e.,
\[
    (\Sigma_{\mathcal D}\cap\Sigma_{\mathrm{prio}})\subseteq\Sigma_{\mathcal P}.
\]
They are \emph{weakly informative for $\mathcal P$ on $\Sigma_{\mathrm{prio}}$} if at least one jointly consistent system possesses $\mathcal P$, i.e.,
\[
    (\Sigma_{\mathcal D}\cap\Sigma_{\mathrm{prio}})\cap\Sigma_{\mathcal P}\neq\emptyset.
\]
\end{definition}

When the prior set is not explicitly mentioned, we adopt the convention $\Sigma_{\mathrm{prio}}=\Sigma$, i.e., no prior constraints are imposed. In this case, the two conditions reduce, respectively, to
\begin{align}
\Sigma_{\mathcal D}&\subseteq\Sigma_{\mathcal P},\label{eq:strong-informativity}\\
\Sigma_{\mathcal D}\cap\Sigma_{\mathcal P}&\neq\emptyset.\label{eq:weak-informativity}
\end{align}
The former is consistent with the existing data informativity notion~\cite{informativity2020}. Hereafter, unless explicitly qualified, ``informative'' means ``strongly informative''.

Strong informativity certifies $\mathcal P$ for all systems in the jointly consistent set $ \Sigma_{\mathcal D}\cap\Sigma_{\mathrm{prio}}$ (thus for the true system); weak informativity certifies that $\mathcal P$ is compatible with the available information. Failure of weak informativity certifies the complementary property $\mathcal P^c$, since
\[
(\Sigma_{\mathcal D}\cap\Sigma_{\mathrm{prio}})\cap\Sigma_{\mathcal P}=\emptyset
\quad\Longleftrightarrow\quad
(\Sigma_{\mathcal D}\cap\Sigma_{\mathrm{prio}})\subseteq\Sigma_{\mathcal P^c}.
\]
Thus, weak but not strong informativity describes an unresolved regime. As illustrated in Fig.~\ref{fig:informativity-relation}, additional trajectories or valid priors shrink the consistent set: this regime may persist or resolve into strong informativity for $\mathcal P$ or $\mathcal P^c$. In this sense, weak informativity describes the transition between strong informativity for $\mathcal P$ and strong informativity for $\mathcal P^c$. Related experimental design questions are studied in~\cite{kang2023minimum}; here, the data are given, and the objective is to determine what can already be concluded from their consistent system set. 

\begin{figure}[t]
\centering
\includegraphics[width=\columnwidth]{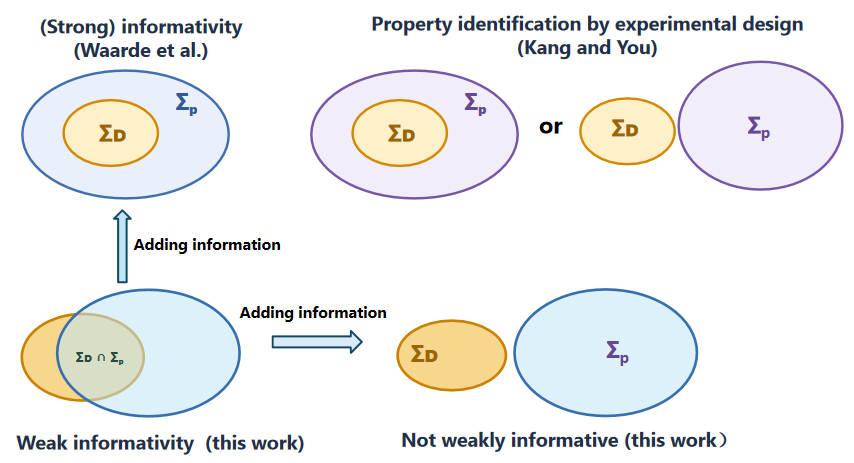}
\caption{Strong and weak data informativity without priors. Additional information shrinks the consistent set; the same interpretation applies to $\Sigma_{\mathcal D}\cap\Sigma_{\mathrm{prio}}$ when priors are available.}
\label{fig:informativity-relation}
\end{figure}

For example, without priors, strong informativity for controllability is equivalent to~\cite{informativity2020}
\begin{equation}\label{rankcondition}
\operatorname{rank}(\lambda X_p-X_f)=n,\qquad\forall\lambda\in\mathbb C.
\end{equation}
Its failure does not determine whether a controllable system remains consistent with the data. Weak informativity addresses this question.
The following two examples illustrate the distinction without priors.

\begin{example}\label{exa:not-weakly-informative}
Consider $n=3$, $m=1$, and the data
\[
X_p=\begin{bmatrix}0&0&0\\1&0&0\\0&0&1\end{bmatrix},\quad
U_p=\begin{bmatrix}0&1&0\end{bmatrix},\quad
X_f=\begin{bmatrix}0&0&0\\0&0&1\\0&1&1\end{bmatrix}.
\]
The data are not informative for controllability because the first row of $\lambda X_p-X_f$ is zero for every $\lambda\in\mathbb{C}$, so condition~\eqref{rankcondition} fails. By solving the linear system of equations $X_f=AX_p+BU_p$, every data-consistent system has the form
\begin{equation} \label{example-para}
    A=\begin{bmatrix}a&0&0\\b&0&1\\c&0&1\end{bmatrix},
\qquad
B=\begin{bmatrix}0\\0\\1\end{bmatrix},
\qquad a,b,c\in\mathbb{R}.
\end{equation}

The first state component is unreachable from the input, so every pair in $\Sigma_{\mathcal D}$ is uncontrollable. Any additional trajectory generated by the same true system produces a new consistent set contained in $\Sigma_{\mathcal D}$; hence, it cannot make the data informative for controllability. Instead, the current data already certify uncontrollability.
\end{example}

\begin{example}\label{exa:additional-data}
Consider the controllable pair
\[
    A_u=\begin{bmatrix}0&1\\0&0\end{bmatrix},
    \qquad
    B_u=\begin{bmatrix}0\\1\end{bmatrix}.
\]
Consider the transition generated from $x(0)=[1,0]^T$ and $u(0)=1$:
$
    x(1)=A_ux(0)+B_uu(0)=[0,1]^T.
$
This transition is also consistent with the uncontrollable pair
$
    A=0_{2\times2},\quad B=[0,1]^T.
$
Consequently, the initial data are weakly but not strongly informative for controllability. Continue the same trajectory with $u(1)=0$ and $u(2)=2$. The aggregated data become
\[
    X_p=\begin{bmatrix}1&0&1\\0&1&0\end{bmatrix},\qquad
    U_p=\begin{bmatrix}1&0&2\end{bmatrix},\qquad
    X_f=\begin{bmatrix}0&1&0\\1&0&2\end{bmatrix}.
\]
The stacked matrix $[X_p^T,U_p^T]^T$ is nonsingular. Hence, these consecutive trajectory data uniquely identify $(A_u,B_u)$ and are strongly informative for controllability.
\end{example}

\begin{problem}\label{prob:weak-informativity}\label{prob:informativity-prior}
Characterize weak informativity for controllability, stabilizability, and stability without priors, and weak and strong informativity for controllability and stability under the affine priors above.
\end{problem}
We investigate the computational limitations of general affine priors and derive explicit criteria for single-affine priors. Section~\ref{sec-noisy} further treats noise in known linear subspaces, where weak informativity can remain useful even when strong informativity fails.
\endgroup

				\section{Weak data informativity without priors}	
				

    	{This section characterizes weak data informativity for controllability, stabilizability, and stability, and establishes a genericity result for controllability over $\Sigma_{{\mathcal D}}$.}
				
				\subsection{Affine parameterization of the data-consistent set}

                {We first parameterize $\Sigma_{{\mathcal D}}$ affinely using the Moore--Penrose pseudoinverse.}
				
				Recall the stacked data matrix
				$Z_p := [X_p^T, U_p^T]^T$.
				Assume that the data are {\emph{consistent}}, namely, there exists at least one pair $(A,B)$ such that
				\(
				X_f = \begin{bmatrix} A & B \end{bmatrix} Z_p.
				\)
				Equivalently, this requires $X_f = X_f Z_p^\dagger Z_p$. Then, all pairs $(A,B)$ consistent with the data can be parameterized as
				\begin{equation}\label{eq:param_general}
					\begin{bmatrix} A_Y & B_Y \end{bmatrix}
					= X_f Z_p^\dagger + Y\bigl(I - Z_p Z_p^\dagger\bigr),
				\end{equation}
				where $Y \in \mathbb{R}^{n\times(n+m)}$ is an arbitrary matrix of free parameters. Consequently, the set of all data-consistent systems can be written as
				\begin{equation}\label{c}
					\Sigma_{{\mathcal D}} = \left\{ (A,B) \;\middle|\;
					\begin{aligned}
						\begin{bmatrix} A & B \end{bmatrix} &= X_f Z_p^\dagger + Y\bigl(I - Z_p Z_p^\dagger\bigr), \\
						Y &\in \mathbb{R}^{n\times(n+m)}
					\end{aligned}
					\right\}.
				\end{equation}

				\subsection{Structural controllability of linearly parameterized systems}
                	The parameterization in \eqref{eq:param_general} reveals that any data-consistent system matrices depend affinely on the free parameter matrix $Y$. Determining whether there exists a controllable pair in $\Sigma_{{\mathcal D}}$ is mathematically equivalent to checking if this affinely parameterized system can be made controllable for some choice of $Y$. This formulation naturally falls into the domain of \emph{structural controllability} for linearly parameterized systems. We therefore recall the framework of \cite{1101135}. 
                    
                    Consider a linear time-invariant system whose state-space matrices depend on a collection of free parameter matrices \(P_i\), \(i\in\mathbf{k}:=\{1,2,\dots,k\}\), in the following linearly parameterized form:
				\begin{equation}
					\label{eq:system_equations}
					\begin{aligned}
						A_p = A_0 + \sum_{i=1}^k B_i P_i C_i,
						B_p = B_0 + \sum_{i=1}^k B_i P_i D_i,
					\end{aligned}
				\end{equation}
				where \(A_0\), \(B_0\), \(B_i\), \(C_i\), and \(D_i\) are fixed real matrices of appropriate dimensions. To rule out triviality, assume that $[C_i,D_i]\ne 0$, $\forall i$.  The pair $(A_p,B_p)$ is said to be structurally controllable if there exist numerical values for the free parameters \(P_i\) ($i\in\mathbf{k}$) such that the resulting system \((A_p,B_p)\) is controllable.


Consider the system $(A_p,B_p)$ given in \eqref{eq:system_equations}.                For $i,j\in\{1,\ldots,k\}$, define
$H_{ij}(\lambda)=C_i(\lambda I-A_0)^{-1}B_j$ and
$H_{i0}(\lambda)=C_i(\lambda I-A_0)^{-1}B_0+D_i$.
The graph $\mathcal{G}$ associated with the parameterization has nodes
$\{0,1,\ldots,k\}$ and contains an edge $j\to i$ if and only if
$H_{ij}(\lambda)$ is not identically zero. For ${\bf s}=\{i_1,\ldots,i_q\}\subseteq
{\bf k}\doteq\{1,\ldots,k\}$, let
$B_{\bf s}=[B_{i_1}\ \cdots\ B_{i_q}]$, and define
\(
C_{\bf s}
=
\begin{bmatrix}
C_{i_1}^{\mathsf T} & \cdots & C_{i_q}^{\mathsf T}
\end{bmatrix}^{\mathsf T},
D_{\bf s}
=
\begin{bmatrix}
D_{i_1}^{\mathsf T} & \cdots & D_{i_q}^{\mathsf T}
\end{bmatrix}^{\mathsf T}.
\)
For ${\bf s}=\emptyset$, these matrices are understood to have compatible
zero dimensions.

\begin{lemma}\cite[Theo 1]{1101135} \label{lemma:sc}
The linearly parameterized system $(A_p,B_p)$ in \eqref{eq:system_equations} is structurally controllable if
and only if the following conditions hold:
\begin{enumerate}
\item for every $i\in{\bf k}$, there exists a directed path
from node $0$ to node $i$ in $\mathcal{G}$;
\item for every $\lambda\in\sigma(A_0)$ (eigenvalues of $A_0$) and every
${\bf s}\subseteq{\bf k}$,
\[
\operatorname{rank}
\begin{bmatrix}
\lambda I-A_0 & -B_0 & B_{\bf s}\\
C_{\bar{\bf s}} & D_{\bar{\bf s}} & 0
\end{bmatrix}\ge n,
\qquad
\bar{\bf s}:={\bf k}\setminus{\bf s}.
\]
\end{enumerate}
\end{lemma}

                				\subsection{Weak data informativity analysis via structural controllability}
				
				By mapping the parameterization \eqref{eq:param_general} into the structural controllability framework (with $Y$ serving as the sole free-parameter block), we can assess the genericity of controllable systems within $\Sigma_{{\mathcal D}}$. Because controllability is characterized by the non-vanishing of polynomial determinants,
                the family \eqref{eq:param_general} has the following ``almost or nothing'' dichotomy. Its proof is standard in the literature (see, e.g., \cite{1101135,zhang2019structural,zhang2024zero}), and  thus omitted here. 
				
				{\bf Fact:}\label{prop:generic_controllability_measure}
					Consider the parameterization of all data-consistent systems
					$(A_Y,B_Y),Y\in\mathbb{R}^{n\times(n+m)},
					$
					given in \eqref{eq:param_general}. Then exactly one of the following two cases occurs:
					
					\medskip
					\noindent
					\textit{(i) Complete uncontrollability:}
					The pair $(A_Y,B_Y)$ is uncontrollable for all $Y\in\mathbb{R}^{n\times(n+m)}$.
					
					\medskip
					\noindent
					\textit{(ii) Generic controllability:}
					There exists $Y_0\in\mathbb{R}^{n\times(n+m)}$ such that $(A_{Y_0},B_{Y_0})$ is controllable. In this case, the set
					\[
					\mathcal U := \left\{ Y\in\mathbb{R}^{n\times(n+m)} \;\middle|\; (A_Y,B_Y) \text{ is uncontrollable} \right\}
					\]
					has Lebesgue measure zero in $\mathbb{R}^{n\times(n+m)}$. Namely, $(A_Y,B_Y)$ is controllable for almost all $Y$.
				

				
				\begin{remark}
                As revealed above,  over the data-consistent set, controllability is either completely impossible or generic. Accordingly,
                although weak informativity is  weaker  than strong informativity, it still guarantees that controllability is the generic situation within the data-consistent set whenever it is possible. Before additional information becomes available, this provides perhaps the most meaningful assessment of controllability that can be drawn from the data-consistent set.
				\end{remark}
				
				The following theorem shows how to computationally verify weak informativity for controllability. The core idea is to check the structural controllability of the linearly parameterized system \eqref{eq:param_general}. The proof is provided in the Appendix. 
				
				\begin{theorem} \label{theo-weak} 
					For consistent data $(U,X)$, let $Z_{p_1}^\dagger \in \mathbb{R}^{T\times n}$ and $Z_{p_2}^\dagger \in \mathbb{R}^{T\times m}$ denote the first $n$ columns and last $m$ columns of $Z_p^\dagger$, respectively. Let $Z_{p_1}^\perp \in \mathbb{R}^{(n+m)\times n}$ and $Z_{p_2}^\perp \in \mathbb{R}^{(n+m)\times m}$ denote the first $n$ columns and last $m$ columns of $Z_{p}^\perp$, respectively, with $Z_{p}^\perp := I - Z_p Z_p^\dagger$. Then, the data $(U,X)$ are weakly informative for controllability if and only if both of the following hold:
					\begin{enumerate}
						\item Either $Z_p^\perp=0$, or the transfer function $Z_{p_2}^\perp + Z_{p_1}^\perp (\lambda I - X_f Z_{p_1}^\dagger)^{-1} X_f Z_{p_2}^\dagger \not\equiv 0$. \footnote{This means the matrix function is not identically zero across all frequencies $\lambda$.}
						\item For all $\lambda \in \mathbb{C}$, the following rank condition holds:
						$$\operatorname{rank}\begin{bmatrix} \lambda I - X_f Z_{p_1}^\dagger & -X_f Z_{p_2}^\dagger \\ Z_{p_1}^\perp & Z_{p_2}^\perp \end{bmatrix} \geq n.$$
					\end{enumerate}
				\end{theorem}
				
				\begin{remark} \label{remark_compare}
					Conditions (1)--(2) of Theorem \ref{theo-weak} are necessary for \eqref{rankcondition}; we verify only Condition (2). Let $A_0=X_fZ_{p_1}^\dagger$, $B_0=X_fZ_{p_2}^\dagger$, $C_1=Z_{p_1}^\perp$, $D_1=Z_{p_2}^\perp$, and $M(\lambda):=\begin{bmatrix}\lambda I-A_0&-B_0\\C_1&D_1\end{bmatrix}$. Since $C_1X_p+D_1U_p=0$ and $A_0X_p+B_0U_p=X_f$,
					\[
					M(\lambda)Z_p=\begin{bmatrix}\lambda X_p-X_f\\0\end{bmatrix}.
					\]
					Thus \eqref{rankcondition} gives $\operatorname{rank}M(\lambda)\ge\operatorname{rank}(M(\lambda)Z_p)=n$, proving Condition (2).
				\end{remark}

				The next corollary gives necessary and sufficient conditions for the data $(U,X)$ to be weakly informative for stabilizability. Note that this corollary is not a trivial extension of Theorem \ref{theo-weak},  whose proof is given in the Appendix. 
				\begin{corollary} \label{coro:weak-stabilizability}
					The consistent data $(U,X)$ are weakly informative for stabilizability if and only if 
					\begin{equation} \label{eq:weak-stabilizability}
						\operatorname{rank}\begin{bmatrix} \lambda I - X_f Z_{p_1}^\dagger & -X_f Z_{p_2}^\dagger \\ Z_{p_1}^\perp & Z_{p_2}^\perp \end{bmatrix} \geq n, \forall |\lambda|\ge 1. 
					\end{equation}
				\end{corollary}
				
				From \cite{informativity2020}, $(U,X)$ are informative for stabilizability if and only if ${\rm rank}(\lambda X_p-X_f)=n$, $\forall |\lambda|\ge 1$. Following Remark \ref{remark_compare}, Condition \eqref{eq:weak-stabilizability} is less restrictive than this condition. 
				
				\begin{theorem} \label{theo-stability}
					The consistent data $(U,X)$ are weakly informative for stability if and only if the pair $(X_f Z_{p_1}^\dagger, Z_{p_1}^\perp)$ is detectable, i.e., 
					\begin{equation} \label{weak-stability}
						\operatorname{rank}\begin{bmatrix} \lambda I - X_f Z_{p_1}^\dagger \\ Z_{p_1}^\perp \end{bmatrix} = n, \quad \forall |\lambda|\ge 1. 
					\end{equation}
				\end{theorem}

                \begin{proof}
                    From the parameterization \eqref{eq:param_general}, $A_Y=X_fZ_{p1}^\dagger+YZ_{p1}^\perp$. Therefore, there exists a $Y$ such that $A_Y$ is stable if and only if $(X_fZ_{p1}^\dagger,Z_{p1}^\perp)$ is detectable, leading to \eqref{weak-stability}.
                \end{proof}

                
				
\begin{example}[Example~\ref{exa:not-weakly-informative}, continued]
{For the data in Example~\ref{exa:not-weakly-informative}, every consistent system is uncontrollable. We now examine stabilizability and stability. Choosing $-1<a<1$ in~\eqref{example-para} yields stabilizable realizations, whereas the fixed eigenvalue $1$ rules out stability for every realization.}

				Now, let us apply Theorem \ref{theo-weak} to the data $(U,X)$. This yields
				\[
				A_0 = \begin{bmatrix} 0 & 0 & 0\\ 0 & 0 & 1\\ 0 & 0 & 1 \end{bmatrix}, 
				B_0 = \begin{bmatrix} 0\\ 0\\ 1 \end{bmatrix}, 
				C_1 = {\begin{bmatrix} 1 & 0 & 0\\ 0 & 0 & 0\\ 0 & 0 & 0\\0&0&0 \end{bmatrix}}, 
				D_1 = {0_{4\times1}},
				\]where $A_0,B_0,C_1,D_1$ are defined in Remark \ref{remark_compare}. 
				Direct calculation verifies that the transfer function $D_1 + C_1 (\lambda I - A_0)^{-1} B_0 \equiv 0$, violating Condition (1) of Theorem \ref{theo-weak}. This confirms that no controllable system exists within $\Sigma_{{\mathcal D}}$. Besides, Condition \eqref{eq:weak-stabilizability} holds, confirming that at least one stabilizable system exists.    Moreover, $(A_0,C_1)$ is undetectable, violating the weak informativity condition for stability.  {\hfill $\square$\par}
                \end{example}
				

				\section{Weak data informativity with priors} \label{sec-prior}
				
				In this section, we investigate weak data informativity for controllability and stability with prior structural knowledge defined in Section~\ref{sec-prob}. 

					\subsection{General affine structural constraints}
					To proceed with our analysis, the following auxiliary result is needed, whose proof is given in the Appendix.

					\begin{proposition} \label{auxiliary_lemma} 
						Let $M_1, \dots, M_r \in \mathbb{R}^{n \times q}$ be a set of given matrices and $\gamma = (\gamma_1, \dots, \gamma_r)^T \in \mathbb{R}^r$ be a vector of given scalars. Consider the linear system of equations for the unknown matrix $W \in \mathbb{R}^{n \times q}$:
						\begin{equation} \label{eq:system} 
							\langle M_i, W \rangle = \gamma_i, \quad i = 1, \dots, r. 
						\end{equation}
						Let $G \in \mathbb{R}^{r \times r}$ be the Gramian matrix defined by $G_{ij} = \langle M_i, M_j \rangle$. Then,
						\begin{enumerate}
							\item \textbf{Consistency:} The system \eqref{eq:system} is consistent if and only if $\gamma \in \operatorname{im}(G)$.
							\item \textbf{General Solution:} If the system is consistent, the general solution is
							\begin{equation} \label{eq:solution}
								W = Y - \sum_{i=1}^r \sum_{j=1}^r (G^\dagger)_{ij} \langle M_j, Y \rangle M_i+\sum_{j=1}^r (G^\dagger \gamma)_j M_j,
							\end{equation}
							where $Y\in\mathbb R^{n\times q}$ is arbitrary.
						\end{enumerate}
					\end{proposition}

					\begin{theorem}[Parameterizations for data and priors]
						\label{thm:structured_parameterization2}
						Consider the set of data-consistent systems $\Sigma_{{\mathcal D}}$. Define  $\tilde M_i:=M_i(Z_p^\perp)^T$ and $\tilde\gamma_i:=\gamma_i-\langle M_i,X_fZ_p^\dagger\rangle$ for $i=1,\dots,r$. Let $\tilde G\in\mathbb R^{r\times r}$ be the Gram matrix with $\tilde G_{ij}:=\langle\tilde M_i,\tilde M_j\rangle$. Then $\Sigma_{{\mathcal D}}\cap\Sigma_{\rm stru}\neq\emptyset$ if and only if $\tilde\gamma\in\operatorname{im}(\tilde G)$. If this condition holds, then the set $\Sigma_{{\mathcal D}}\cap \Sigma_{\rm stru}$ can be parameterized as
											\begin{equation} \label{para_data_prior}
						\Sigma_{{\mathcal D}} \cap \Sigma_{\mathrm{\rm stru}} = \left\{ \Gamma_0 + \Omega(Y) Z_p^\perp \;\middle|\; Y \in \mathbb{R}^{n\times (n+m)} \right\},
					\end{equation}
					where $\Gamma_0 := X_f Z_p^\dagger + \sum_{j=1}^r (\tilde{G}^\dagger \tilde{\gamma})_j \tilde{M}_j Z_p^\perp$ and 
					\[
					\Omega(Y) := Y - \sum_{i=1}^r \sum_{j=1}^r (\tilde{G}^\dagger)_{ij} \langle \tilde{M}_j, Y \rangle \tilde{M}_i.
					\]
					\end{theorem}

					\begin{proof}
						Substituting the data-based parameterization from \eqref{eq:param_general} into the structural constraint \eqref{eq:structural_constraint} yields:
						\(
						\langle M_i, X_f Z_p^\dagger + Y Z_p^\perp \rangle = \gamma_i,\,\, i=1,...,r.
						\)
						By distributing the terms, we obtain:
						\[
						\langle M_i,YZ_p^\perp \rangle=\gamma_i-\langle M_i,X_fZ_p^\dagger \rangle,\,\, i=1,...,r.
						\]
						
						Using the Frobenius inner product, $\langle M_i, YZ_p^\perp \rangle = \operatorname{Tr}(M_i^T Y Z_p^\perp) = \operatorname{Tr}(Z_p^\perp M_i^T Y) = \operatorname{Tr}((M_i (Z_p^\perp)^T)^T Y)=\langle M_i(Z_p^\perp)^T,Y\rangle.$ Thus, with $\tilde M_i:=M_i(Z_p^\perp)^T$ and $\tilde\gamma_i:=\gamma_i-\langle M_i,X_fZ_p^\dagger\rangle$, the structural constraints are equivalent to the standard system
						\begin{equation} \label{eq:standard_form}
							\langle \tilde{M}_i,Y\rangle=\tilde{\gamma}_i,
							\quad i=1,\dots,r.
						\end{equation}
						
						By Proposition \ref{auxiliary_lemma}, the system \eqref{eq:standard_form} is consistent if and only if $\tilde\gamma\in\operatorname{im}(\tilde G)$. When this condition holds, the set of all solutions for $Y$ is given by \eqref{eq:solution}.
						Substituting this general solution for $Y$ back into the original data-consistent parameterization \eqref{eq:param_general} yields the parameterization of $\Sigma_{{\mathcal D}}\cap \Sigma_{\rm stru}$. 
					\end{proof}


									The parametrization in \eqref{para_data_prior} describes the set of systems $(A, B)$ that are consistent with the available data and structural priors. Since each entry of $(A,B)$ in $\Sigma_{{\mathcal D}}\cap \Sigma_{\rm stru}$ depends affinely on $Y$,  it is easy to see that if this set contains one controllable system, then almost all systems consistent with the data and the prior structure, i.e., almost all systems in $\Sigma_{{\mathcal D}}\cap \Sigma_{\rm stru}$, are controllable. 
						
						Let \(\{E_{ij}\}_{1\le i\le n,\;1\le j\le n+m}\) denote the standard basis of \(\mathbb{R}^{n\times (n+m)}\).
						Since $\Omega(Y)$ is linear in $Y$, expanding $Y = \sum_{i=1}^n \sum_{j=1}^{n+m} \mu_{ij} E_{ij}$ in \eqref{para_data_prior} yields
					\begin{equation}\label{matrix_net}
						[A, B] = \Gamma_0 + \sum_{i=1}^n \sum_{j=1}^{n+m} \mu_{ij} \, \Omega(E_{ij}) Z_p^{\perp},
					\end{equation}
					where $\mu_{ij} \in \mathbb{R}$ are free scalar parameters. Expression \eqref{matrix_net} forms a \emph{matrix net} in the sense of \cite{anderson1982structural}.
					 We summarize our result in the following proposition, whose proof is apparent based on Theorem \ref{thm:structured_parameterization2}. 
					\begin{proposition}
						The consistent data $(U,X)$ are weakly informative for controllability on the prior set $\Sigma_{\rm stru}$ if and only if 
						$\tilde \gamma\in {\rm im}(\tilde G)$ and the system described by the matrix net \eqref{matrix_net} is structurally controllable.
					\end{proposition}

					Generally, checking the structural controllability of a matrix net incurs exponentially increasing complexity in the number of free parameters $\mu_{ij}$ if their coefficient matrices $\Omega({E_{ij}})Z_p^\perp$ are not of rank-one \cite{anderson1982structural}. 
					Similarly, data $(U,X)$ are weakly informative for stability on the prior set $\Sigma_{\mathrm{\rm stru}}$ if and only if $\tilde \gamma \in \operatorname{im}(\tilde G)$ and there exists a stable system realization within \eqref{para_data_prior}. However, the following result reveals that certifying weak data informativity for stability under general affine priors is computationally intractable.

                    \begin{proposition}
                        Verifying whether consistent data $(U,X)$ are weakly informative for stability on a prior set $\Sigma_{\rm stru}$ is NP-hard.
                    \end{proposition}
                    \begin{proof}
                                          When no data are available (e.g., $U = 0, X = 0$, leading to $\Sigma_{{\mathcal D}} = \Sigma$), this problem contains the Static Output Feedback stabilization (SOF) problem\footnote{The SOF problem seeks to determine if, given matrices $A_0 \in \mathbb{R}^{n \times n}$, $B_0 \in \mathbb{R}^{n \times m}$, and $C_0 \in \mathbb{R}^{p \times n}$, there exists a feedback gain $K \in \mathbb{R}^{m \times p}$ such that the closed-loop matrix $A_0 + B_0 K C_0$ is Schur stable.} as a special case, which has recently been shown to be NP-hard \cite{lofberg2026static,ahmadi2026complexity}. 
                    \end{proof}

					These computational barriers motivate us to focus on the structured single-affine prior constraint introduced in Assumption \ref{assump:structured_block}. 

					\subsection{Single-affine prior knowledge} \label{sec-block-zero}
							We now investigate weak data informativity under single-affine constraints \eqref{eq:Block_structural_constraint}. The following proposition establishes the parameterization of data-consistent systems in this case.

					\begin{proposition}
						\label{thm:block_structured_param}
						Consider consistent data $(U,X)$ and a single-affine constraint with $S_{\ell}$ and $S_r$. Define
						\[
						\Sigma_0:=S_\ell X_f Z_p^\dagger S_r-Q,
						\Sigma_1:=S_\ell,
						\Sigma_2:=\bigl(I-Z_pZ_p^\dagger\bigr)S_r.
						\]
						Then the set \(\Sigma_{{\mathcal D}}\cap\hat{\Sigma}_{{\rm stru}}\) is nonempty if and only if
						\begin{equation}
							\label{eq:block_structured_solvability}
							\Sigma_1\Sigma_1^\dagger \Sigma_0 \Sigma_2^\dagger \Sigma_2=\Sigma_0.
						\end{equation}
						In this case, the set of all systems \((A,B)\) in $\Sigma_{{\mathcal D}}\cap\hat{\Sigma}_{\rm stru}$ is given by
				\begin{equation}\label{parA_uingle_a}
    				\left\{ (A,B) \;\middle|\;
					\begin{aligned}
						&[A,B] = X_fZ_p^\dagger
							+\Sigma_1^\dagger(-\Sigma_0)\Sigma_2^\dagger Z_p^\perp
							+\Sigma_{1\perp} U Z_p^\perp
							\\ &+V\Sigma_2^\perp Z_p^\perp,\quad 
						U,V \in \mathbb{R}^{n\times(n+m)}
					\end{aligned}
					\right\}.
\end{equation}
					\end{proposition}
					
					\begin{proof}
	{Substitution of \eqref{eq:param_general} into $S_\ell[A\,B]S_r=Q$ gives
						\begin{equation}
							\label{eq:block_matrix_eq}
							\Sigma_1 Y \Sigma_2=-\Sigma_0.
						\end{equation}
						The standard solvability criterion for this two-sided equation is precisely \eqref{eq:block_structured_solvability}; when it holds, all solutions are
					\begin{equation} \label{expresssion_Y}
						Y=\Sigma_1^\dagger(-\Sigma_0)\Sigma_2^\dagger+\Sigma_{1\perp} U+V\Sigma_2^\perp,
						\end{equation}
						where $U,V$ are arbitrary. Substitution in \eqref{eq:param_general} gives \eqref{parA_uingle_a}.}
                    	\end{proof}
					

					
					By defining
					\(
					W_0:=X_f Z_p^\dagger+\Sigma_1^\dagger(-\Sigma_0)\Sigma_2^\dagger Z_p^\perp,
					S_1:=U, S_2:=V, G_1:=\Sigma_{1\perp}, G_2:=I, H_1:=Z_p^\perp,H_2:=\Sigma_2^\perp Z_p^\perp,
					\)
					any member $(A,B)$ in $\Sigma_{{\mathcal D}}\cap\hat{\Sigma}_{\rm stru}$ can be parameterized by
					\begin{equation}
						\label{eq:block_structured_rank_one_form}
						[A\,\, B]
						=
						W_0+\sum_{i=1}^2 G_i S_i H_i,
					\end{equation}where $S_i$ are arbitrary matrices of compatible dimensions.  Assessing whether data are weakly informative for controllability on $\hat \Sigma_{\mathrm{\rm stru}}$ is thus equivalent to testing structural controllability for \eqref{eq:block_structured_rank_one_form}, where Lemma \ref{lemma:sc} provides a criterion that can be verified in polynomial time \cite{zhang2019structural}.  Unlike the general matrix net \eqref{matrix_net}, each free element in $S_i$ within \eqref{eq:block_structured_rank_one_form} enters with rank-one coefficient matrices. This differentiates the computational complexity of the single-affine case from that of the general case.   

                    Furthermore, one can suitably choose $U,X$ such that $\Sigma_2^\perp Z_p^\perp=0$ while $\Sigma_{1\perp}$ and $Z_p^\perp$ in \eqref{parA_uingle_a} are general enough. In this way, one can conclude that
					the problem of verifying whether the data $(U,X)$ are weakly informative for stability on the prior set $\hat \Sigma_{\rm stru}$ is at least as hard as the SOF problem. 	

					\section{Strong data informativity with priors}
					In this section, we consider strong data informativity with prior structural knowledge for controllability, stabilizability, and stability.

					\begin{proposition}\label{hard}
						The problem of verifying whether data $(U,X)$ are informative for controllability on the prior set $\Sigma_{\rm stru}$ is NP-hard. 
					\end{proposition}
					
					\begin{proof}
						Consider the case $U=0,X=0$. Then, $\Sigma_{{\mathcal D}}=\Sigma$, and the informativity problem becomes determining whether all systems in the affinely parameterized set $\Sigma_{\rm stru}$ are controllable. This problem has been shown to be NP-hard in \cite[Theo 1]{zhang2023real}.
					\end{proof}
					
					In the case $U=0,X=0$, verifying data informativity for stability on the set $\Sigma_{\rm stru}$ reduces to the problem of determining whether all systems in the affinely parameterized set $\Sigma_{\rm stru}$ are stable. There is no known deterministic polynomial-time algorithm for this problem \cite{nemirovskii1993several}. {Below, we focus on single-affine prior knowledge. For controllability, we give an exact characterization and sufficient rank conditions, which become necessary and sufficient when $S_\ell=I_n$.}

\begingroup
\begin{theorem}\label{theo:informativity-controllability}
Suppose $\Sigma_{\mathcal D}\cap\hat\Sigma_{\rm stru}\ne\emptyset$, and let $\Sigma_0,\Sigma_1,\Sigma_2$ be as in Proposition~\ref{thm:block_structured_param}. Define
$E_n=[I_n\ \ 0_{n\times m}],\quad M_0=S_\ell^\dagger(-\Sigma_0)\Sigma_2^\dagger$,
$R=\Sigma_2^\perp Z_p^\perp,\quad W=I-R^\dagger R$,
and $\mathcal K_\lambda=\{v\in\mathbb C^n\mid v^*(\lambda X_p-X_f)=0\}$, where $*$ denotes conjugate transpose. For $v\in\mathcal K_\lambda$, set
$\mathcal L(v)=\begin{bmatrix}
(Z_p^\perp)^T\otimes(v^*S_{\ell\perp})&R^T\otimes v^*
\end{bmatrix},\quad
b(\lambda,v)=\bigl[v^*(\lambda E_n-M_0)Z_p^\perp\bigr]^T$,
where $\otimes$ denotes the Kronecker product, and define the real matrices
\[
\mathcal L_{\mathbb R}(v)=\begin{bmatrix}\operatorname{Re}\mathcal L(v)\\\operatorname{Im}\mathcal L(v)\end{bmatrix},\qquad
b_{\mathbb R}(\lambda,v)=\begin{bmatrix}\operatorname{Re}b(\lambda,v)\\\operatorname{Im}b(\lambda,v)\end{bmatrix},
\]where ${\rm Re}(\cdot)$ and ${\rm Im}(\cdot)$ denote the real and imaginary parts, respectively. 
The data $(U,X)$ are informative for controllability on $\hat\Sigma_{\rm stru}$ if and only if
\begin{equation}\label{eq:exact-strong-controllability}
\operatorname{rank}\begin{bmatrix}\mathcal L_{\mathbb R}(v)&b_{\mathbb R}(\lambda,v)\end{bmatrix}
>\operatorname{rank}\mathcal L_{\mathbb R}(v)
\end{equation}
for every $\lambda\in\mathbb C$ and $v\in\mathcal K_\lambda\setminus\{0\}$. In particular, the following two conditions, imposed for every $\lambda\in\mathbb C$, are sufficient:
\begin{enumerate}
\item[1.] $\operatorname{rank}\begin{bmatrix}\lambda X_p-X_f&S_{\ell\perp}\end{bmatrix}=\operatorname{rank}(\lambda X_p-X_f)$;
\item[2.] $\operatorname{rank}\begin{bmatrix}\lambda X_p-X_f&(\lambda E_n-M_0)Z_p^\perp W\end{bmatrix}=n$.
\end{enumerate}
When $S_\ell=I_n$, Condition~2 alone is necessary and sufficient for data $(U,X)$ to be informative for controllability on $\hat\Sigma_{\rm stru}$.
\end{theorem}

\begin{proof}
    The proof is given in the Appendix. 
\end{proof}

Condition~\eqref{eq:exact-strong-controllability} retains the restriction to real system matrices. Although each fixed $(\lambda,v)$ requires only a real linear consistency test, the condition quantifies over all complex eigenvalues $\lambda$ and candidate left vectors $v$ in ${\mathcal K}_{\lambda}$ and does not itself provide an efficient finite verification procedure. Conditions~1--2 give a sufficient alternative. 
\endgroup

{When $S_{\ell}=0, Q=0$, Conditions 1 and 2 collapse to ${\rm rank}(\lambda X_p-X_f)=n$, consistent with the strong data informativity condition without priors. When $S_\ell=I_n$, Condition~2 reduces to
						$
				\operatorname{rank}[
						\lambda X_p-X_f, \lambda E_nS_r-Q
					]=n,\, \forall\lambda\in\mathbb C.
						$
						Compared with the conventional condition $\operatorname{rank}(\lambda X_p-X_f)=n$, the additional block contributed by the prior may supply the missing ranks. Hence, strong informativity can be satisfied from fewer data, possibly $T<n$.}

\begin{example}
Consider the discrete-time linear system obtained from the
inverted-pendulum model given in \cite[Example 3--6]{Modern_Control_Ogata}
with sampling period $T_s=0.1$, whose system matrices $(A_u,B_u)$ in \eqref{eq:sys} are given by
\[
A_u=
\begin{bmatrix}
1 & T_s & 0 & 0\\
T_s\dfrac{(M+m)g}{Ml} & 1 & 0 & 0\\
0 & 0 & 1 & T_s\\
-T_s\dfrac{mg}{M} & 0 & 0 & 1
\end{bmatrix},
\qquad
B_u=
\begin{bmatrix}
0\\
-\dfrac{T_s}{Ml}\\
0\\
\dfrac{T_s}{M}
\end{bmatrix}.
\]
Here, $m$ and $M$ are the masses of the pendulum bob and the cart,
$l$ is the length of the massless rod, and $g$ is the acceleration
of gravity. Suppose that the fourth column of $A_u$ is available as
prior structural knowledge, i.e.,
$A_u(:,4)=[0~0~T_s~1]^\top$.
This prior knowledge is determined by the kinematic structure of the
inverted-pendulum model and does not require the exact values of the
physical parameters $M$, $m$, $l$, and $g$.

The above prior information can be represented in the single-affine
form $S_\ell[A~B]S_r=Q$, where
\[
S_\ell=I_4,
S_r=
\begin{bmatrix}
0&0&0&1&0
\end{bmatrix}^{\top},
Q=
\begin{bmatrix}
0&0&T_s&1
\end{bmatrix}^{\top}.
\]

In the simulation, we choose $M=1$, $m=0.2$, $g=9.8$, $l=0.5$. We collect only $T=3$ input-state samples, yielding
\[
U_p=
\begin{bmatrix}
1&0&-1
\end{bmatrix},
X_p=
\begin{bmatrix}
0 & -0.1 & -0.22\\
-1 & -1.2 & -1.4352\\
0 & 0.05 & 0.11\\
0.5 & 0.6 & 0.6196
\end{bmatrix},
\]
and
\[
X_f=
\begin{bmatrix}
-0.1 & -0.22 & -0.3635\\
-1.2 & -1.4352 & -1.7526\\
0.05 & 0.11 & 0.172\\
0.6 & 0.6196 & 0.5627
\end{bmatrix}.
\]

Under the conventional data informativity framework, the condition
$\operatorname{rank}(\lambda X_p-X_f)=4$ for all
$\lambda\in\mathbb{C}$ cannot be satisfied, since
$\operatorname{rank}(\lambda X_p-X_f)\leq T=3<4$.
Therefore, the collected data alone are not informative for
controllability.

We next show that the system matrices cannot be uniquely determined
even after incorporating the prior knowledge. Since $\operatorname{rank}([X_p^\top~U_p^\top]^\top)=3$ and
$\operatorname{rank}([Z_p~S_r])=4<5$, the data together with
the prior knowledge are still insufficient to uniquely determine
$[A~B]$. We now apply Theorem~\ref{theo:informativity-controllability}. Since $S_\ell=I_4$, we have
$S_\ell^\perp=0$, and hence Condition~1 of
Theorem~\ref{theo:informativity-controllability} is automatically
satisfied. For Condition~2, a direct calculation shows that a $4\times4$ minor of
$
\begin{bmatrix}
\lambda X_p-X_f&
(\lambda E_4-M_0)Z_p^\perp W
\end{bmatrix}
$
has a nonzero determinant for all $\lambda\in\mathbb{C}$. Hence, the
augmented matrix has full row rank, and Condition~2 of
Theorem~\ref{theo:informativity-controllability} is satisfied.\qed 
\end{example}

\begingroup
The same argument yields a stabilizability counterpart. The proof is omitted due to its obviousness. 
\begin{corollary}
Under the assumptions and notation of Theorem~\ref{theo:informativity-controllability}, the data $(U,X)$ are informative for stabilizability on $\hat\Sigma_{\rm stru}$ if and only if~\eqref{eq:exact-strong-controllability} holds for every $\lambda\in\mathbb C$ with $|\lambda|\ge1$ and every $v\in\mathcal K_\lambda\setminus\{0\}$. Conditions~1--2 of that theorem, imposed only for $|\lambda|\ge1$, are sufficient. When $S_\ell=I_n$, they are also necessary, and Condition~1 holds automatically.
\end{corollary}
\endgroup

Before proceeding to informativity for stability with structural prior knowledge, we first present the result without prior knowledge. 
					
					\begin{proposition} \label{stable-no-prior}
						Let $E^T_n = \begin{bmatrix} I_n \\ 0_{m \times n} \end{bmatrix} \in \mathbb{R}^{(n+m) \times n}$. The consistent data $(U,X)$ are informative for stability if and only if the following two conditions hold:
						\begin{enumerate}
							\item $\operatorname{rank} \begin{bmatrix} X_p \\ U_p \end{bmatrix} = n + \operatorname{rank}(U_p)$;
							\item The matrix $A_0 = X_f Z_p^\dagger E_n^T$ is Schur stable.
						\end{enumerate}
					\end{proposition}
					
					\begin{proof}
	By multiplying the parameterization \eqref{eq:param_general} from the right by $E_n^T$, we extract $A_Y = \begin{bmatrix} A_Y & B_Y \end{bmatrix} E_n^T = X_f Z_p^\dagger E_n^T + Y (I - Z_p Z_p^\dagger) E_n^T$.
	Let $A_0 = X_f Z_p^\dagger E_n^T$ and $C = (I - Z_p Z_p^\dagger) E_n^T$. The expression simplifies to $A_Y = A_0 + YC$.
	
	\textbf{Necessity:} Suppose $C \neq 0$. Then there exists a non-zero vector $v \in \mathbb{R}^n$ such that $w = Cv \neq 0$. Since $Y$ is arbitrary, we can choose a rank-1 matrix $Y$ such that $Yw = (\lambda I - A_0)v$ for any $\lambda \in \mathbb{R}$. Consequently, $A_Y v = (A_0 + YC)v = A_0 v + Yw = \lambda v$. By choosing $\lambda$ such that $|\lambda| \ge 1$, the matrix $A_Y$ fails to be Schur stable. Thus, $C=0$ is a necessary condition. 
	
	The condition $C = (I - Z_p Z_p^\dagger) E_n^T = 0$ implies that the columns of $E_n^T$ must lie in the range of $Z_p$. Geometrically, this is equivalent to:
	\begin{equation}
		\label{eq:rank_equiv}
		\operatorname{rank} \begin{bmatrix} Z_p & E_n^T \end{bmatrix} = \operatorname{rank}(Z_p).
	\end{equation}
	Expanding the left-hand side of \eqref{eq:rank_equiv} using the definitions of $Z_p$ and $E_n^T$:
	\begin{equation*}
		\operatorname{rank} \begin{bmatrix} X_p & I_n \\ U_p & 0 \end{bmatrix}= n + \operatorname{rank}(U_p).
	\end{equation*}
	Thus, $C=0$ is equivalent to Condition 1. If $C=0$, then $A_Y = A_0$ for all $Y$. For $A_0$ to be Schur stable, Condition 2 must hold.
	
	\textbf{Sufficiency:} If Condition 1 holds, then $C = (I - Z_p Z_p^\dagger) E_n^T = 0$. Substituting this into the expression for $A_Y$, we obtain $A_Y = A_0$ for all $Y \in \mathbb{R}^{n \times (n+m)}$. If Condition 2 holds, $A_0$ is Schur stable, and therefore $A_Y$ is Schur stable for every choice of $Y$.
\end{proof}
					
					\begin{remark}
						Proposition \ref{stable-no-prior} extends Corollary 11 of \cite{informativity2020}, which considered the stability of an autonomous system $x(t+1)=Ax(t)$.
						In this case, $U_p$ does not exist ($m=0$), $Z_p=X_p$, and $E_n=I_n$. Then, conditions (1) and (2) reduce to that $X_p$ has full row rank and $X_fX_p^{\dagger}$ is stable, which coincides with \cite{informativity2020}. Note that in this case, the system matrices $A$ can be uniquely identified.  For $m>0$, $A$ is unique, whereas $B$ need not be.
					\end{remark}

                    We now consider informativity for stability with the single-affine prior knowledge. 

            \begin{theorem} \label{stability_prior}
	Given the consistent data $(U,X)$ and the prior set $\hat \Sigma_{\rm stru}$, let $\Sigma_0$ and $\Sigma_2$ be defined as in Proposition \ref{thm:block_structured_param}  and  $M_0=S_{\ell}^\dagger(-\Sigma_{0})\Sigma_2^\dagger$. Assume that $\Sigma_{{\mathcal D}}\cap \hat \Sigma_{\rm stru}\ne \emptyset$. 	Define matrices $A_0 \in \mathbb{R}^{n \times n}$ and  $C \in \mathbb{R}^{(n+m) \times n}$ as:
	\begin{align*}
		A_0 &= (X_f Z_p^\dagger + M_0(I - Z_p Z_p^\dagger)) E_n^T, \\
		C &= (I - Z_p Z_p^\dagger) E_n^T.
	\end{align*}
	Then, the data $(U,X)$ are informative for stability on the prior set $\hat \Sigma_{\rm stru}$ if and only if the following three conditions hold:
	\begin{enumerate}
		\item $A_0$ is Schur stable;
		\item $\Sigma_{2}^\perp C = 0$;
		\item $C A_0^k S_{\ell \perp} = 0$ for all $k = 0, 1, \dots, n-1$.
	\end{enumerate}
\end{theorem}

\begin{proof}
	From the proof of Proposition \ref{thm:block_structured_param}, $(A_Y,B_Y)\in \hat \Sigma_{\rm stru}\cap \Sigma_{{\mathcal D}}$ if and only if
	there exists a matrix $Y$, written as \eqref{expresssion_Y}, such that $[A_Y,B_Y]$ is parameterized as \eqref{eq:param_general}.	
	Substituting \eqref{expresssion_Y} into \eqref{eq:param_general}, we expand $A_Y$ as follows:
	\begin{align*}
		A_Y &= X_f Z_p^\dagger E_n^T + (M_0 + S_{\ell \perp} U + V \Sigma_{2}^\perp)(I - Z_p Z_p^\dagger) E_n^T \\
		&= (X_f Z_p^\dagger E_n^T + M_0 C) + S_{\ell \perp} U C + V \Sigma_{2}^\perp C \\
		&= A_0 + S_{\ell \perp} U C + V \Sigma_{2}^\perp C.
	\end{align*}
	
	\textbf{Necessity:} 
	Setting $U = 0$ and $V = 0$, we have $A_Y = A_0$. For $A_Y$ to be stable for all $U, V$, $A_0$ must be Schur stable, satisfying Condition (1). Next, suppose $\Sigma_{2}^\perp C \neq 0$. There exists a vector $v$ such that $w = \Sigma_{2}^\perp C v \neq 0$. Since $V$ is arbitrary, we can choose $V$ such that $V w = \gamma v - A_0 v$. Setting $U=0$, we get $A_Y v = \gamma v$. For $|\gamma| \ge 1$, stability is lost. Thus, $\Sigma_{2}^\perp C = 0$ is necessary, satisfying Condition (2).
	With $\Sigma_{2}^\perp C = 0$, $A_Y = A_0 + S_{\ell \perp} U C$. Then $\det(zI - A_Y) = \det(zI - A_0) \cdot \det\left(I - H(z) U\right)$ for any $z\in \mathbb C$ with $\det(zI-A_0)\ne 0$, where the transfer function $H(z) = C (zI - A_0)^{-1} S_{\ell \perp}$.  For the eigenvalues to be independent of $U$, the transfer function $H(z)$ must be identically zero; otherwise, for any $z_0\in \mathbb R$ with $z_0>1$ and $H(z_0)\ne 0$, there exists a real $U$ making $\det(I-H(z_0)U)=0$, making $z_0$ an unstable eigenvalue of $A_Y$. Expanding the resolvent $(zI - A_0)^{-1}$ into a power series, this requires $C A_0^k S_{\ell \perp} = 0$ for all $k \ge 0$, which is Condition (3).
	
	\textbf{Sufficiency:}
	If Conditions (2) and (3) hold, then $V \Sigma_{2}^\perp C = 0$ and the term $S_{\ell \perp} U C$ does not affect the characteristic polynomial of $A_0$ because:
	\begin{equation*}
		\det(zI - A_Y)\!\! =\!\! \det(zI - A_0) \det(I - U H(z)) \!=\!\! \det(zI - A_0).
	\end{equation*}
	Since $A_0$ is Schur stable by Condition (1), all eigenvalues of $A_Y$ lie within the unit circle for all $U, V$.
\end{proof}

\begin{remark}
    From the proof of Proposition \ref{stable-no-prior}, without priors, the data $(U,X)$ are informative for stability if and only if $C=0$ and $A_0$ is stable. These conditions are indeed relaxed when more prior knowledge is utilized, as shown in Theorem \ref{stability_prior}.  
\end{remark}

                \section{Extensions with noisy data} \label{sec-noisy}

Consider system \eqref{eq:sys} affected by process noise, given by
\begin{equation} \label{sys-noise}
	x(t+1) = A_u x(t) + B_u u(t) + d(t),
\end{equation} 
where $d(t) \in \mathbb{R}^n$ is an unknown disturbance (or noise) signal. Let $D_p = \begin{bmatrix} d(0) & d(1) & \cdots & d(T-1) \end{bmatrix}.$
With $X_p, X_f$, and $U_p$ defined as in Section \ref{sec-prob}, the corresponding data equation becomes
$
X_f = A X_p + B U_p + D_p.
$
Assume that the noise sequence $D_p$ belongs to a prior set $\mathbb{D} \subseteq \mathbb{R}^{n\times T}$, defined as
\[
\mathbb{D} := \left\{ D \in \mathbb{R}^{n\times T} \mid \operatorname{im}(D) \subseteq \mathcal{D} \right\},
\]
where $\mathcal{D} \subseteq \mathbb{R}^n$ is a known linear subspace. This scenario naturally arises when the noise stems from unknown sparse disturbances affecting only a specific subset of the state variables (e.g., $d(t) = F e(t)$ for a known matrix $F$ and an unknown signal $e(t) \in \mathbb{R}^q$).

The set of all systems that can explain the data $(U, X)$ for some allowable noise $D_p \in \mathbb{D}$ is given by
\[
\Sigma_{{\mathcal D}}^n = \left\{ (A,B) \mid \exists D_p \in \mathbb{D} \text{ s.t. } X_f = A X_p + B U_p + D_p \right\}.
\]
Note that without prior knowledge on $\mathbb{D}$ (i.e., if $\mathbb{D} = \mathbb{R}^{n\times T}$), informativity analysis becomes impossible, as any pair $[A,B] \in \mathbb{R}^{n\times (n+m)}$ could trivially explain the data by setting $D_p = X_f - A X_p - B U_p$.

Let $E \in \mathbb{R}^{r\times n}$ be a full-row-rank matrix with the maximum possible number of rows such that $ED = 0$ for all $D \in \mathbb{D}$. Left-multiplying the data equation by $E$ yields
\begin{equation} \label{noise-equation}
	E X_f = E [A \,\, B] \begin{bmatrix} X_p \\ U_p \end{bmatrix} = E [A \,\, B] Z_p.
\end{equation}
Assuming the consistency condition $E X_f Z_p^\perp = 0$ holds, the general solution for $[A,B]$ satisfying \eqref{noise-equation} is parameterized as
\begin{equation} \label{parameterize-noisy} 
	[A, B] = E^\dagger E X_f Z_p^\dagger + E_\perp U + V Z_p^\perp,
\end{equation}
where $U$ and $V$ are arbitrary matrices of compatible dimensions.

\begin{remark}
	We have adopted the same noise setting as in \cite{eising2021informativity}. For stability analysis, it is possible to extend our results to general quadratic bound constraints as in \cite{formulas2020}, which is left for future work.
\end{remark}

We say the data $(U,X)$ are \emph{weakly informative for property $\mathcal{P}$ on the noise set $\mathbb{D}$} if $\Sigma_{{\mathcal D}}^n \cap \Sigma_{\mathcal{P}} \ne \emptyset$. Similarly, the data $(U,X)$ are \emph{strongly informative for property $\mathcal{P}$ on $\mathbb{D}$} if $\Sigma_{{\mathcal D}}^n \subseteq \Sigma_{\mathcal{P}}$.

\begin{corollary} \label{weak-controllability}
	Suppose the data $(U,X)$ and the noise set $\mathbb{D}$ satisfy the consistency condition $E X_f Z_{p\perp} = 0$. The data $(U,X)$ are weakly informative for controllability on $\mathbb{D}$ if and only if the following two conditions hold:
	\begin{enumerate}
		\item $Z_{p_1}^\perp (\lambda I - E^\dagger E X_f Z_{p_1}^\dagger)^{-1} [E^\dagger E X_f Z_{p_2}^\dagger, E_\perp] + [Z_{p_2}^\perp, 0] \not\equiv 0$.
		\item For all $\lambda \in \mathbb{C}$,
		\[
		\operatorname{rank} \begin{bmatrix} 
			\lambda I - E^\dagger E X_f Z_{p_1}^\dagger & -E^\dagger E X_f Z_{p_2}^\dagger & E_\perp \\ 
			Z_{p_1}^\perp & Z_{p_2}^\perp & 0 
		\end{bmatrix} \ge n,
		\]
	\end{enumerate}
	where $Z_{p_1}^\dagger$ (resp. $Z_{p_1}^\perp$) and $Z_{p_2}^\dagger$ (resp. $Z_{p_2}^\perp$) denote the first $n$ columns and the last $m$ columns of $Z_p^\dagger$ (resp. $Z_p^\perp$), respectively.
\end{corollary}

\begin{proof}
	By definition, the consistent data $(U,X)$ are weakly informative for controllability on the noise set $\mathbb{D}$ if and only if the linearly parameterized system \eqref{parameterize-noisy} is structurally controllable. By mapping \eqref{parameterize-noisy} into the standard form $A_P = A_0 + \sum_{i=1}^{2} B_i P_i C_i$ and $B_P = B_0 + \sum_{i=1}^{2} B_i P_i D_i$ using the assignments $A_0 = E^\dagger E X_f Z_{p_1}^\dagger$, $B_0 = E^\dagger E X_f Z_{p_2}^\dagger$, $B_1 = E_\perp$, $C_1 = \begin{bmatrix} I_n \\ 0 \end{bmatrix}$, $D_1 = \begin{bmatrix} 0 \\ I_m \end{bmatrix}$, $B_2 = I_n$, $C_2 = Z_{p_1}^\perp$, and $D_2 = Z_{p_2}^\perp$, we can directly apply \cite[Theorem 1]{1101135} to $(A_P, B_P)$ to obtain the simplified conditions as required.
\end{proof}

\begin{corollary} \label{weak-stability-noise-corollary}
	The data $(U,X)$ are weakly informative for stability on the noise set $\mathbb{D}$ if and only if $E X_f Z_{p\perp}= 0$ (consistency condition) and
	\begin{equation}\label{weak-stability-noisy}
		\operatorname{rank} \begin{bmatrix} E^\dagger E X_f Z_{p_1}^\dagger - \lambda I & E_\perp \\ Z_{p_1}^\perp & 0 \end{bmatrix} \ge n, \quad \forall |\lambda| \ge 1,
	\end{equation}
	where $Z_{p_1}^\perp$ and $Z_{p_1}^\dagger$ are defined as in Theorem \ref{theo-weak}.
\end{corollary}

\begin{proof}
					Define $A_0=E^\dagger E X_f Z_{p1}^\dagger$, $B_1=E_\perp$, $C_1= Z_{p1}^\perp$. Then, $A=A_0 + B_1 U + V C_1$.
					
					\textbf{Necessity:} 
					Assume there exist $U$ and $V$ such that $A = A_0 + B_1 U + V C_1$ is Schur stable. For any $\lambda \in \mathbb{C}$ with $|\lambda| \ge 1$, the matrix $A - \lambda I$ must be non-singular. We apply invertible block transformations to the system matrix:
					\begin{align*}
						&\begin{bmatrix} I_n & V \\ 0 & I_{n+m} \end{bmatrix} \begin{bmatrix} A_0 - \lambda I & B_1 \\ C_1 & 0 \end{bmatrix} \begin{bmatrix} I_n & 0 \\ U & I_n \end{bmatrix}\\
						& = \begin{bmatrix} A_0 + B_1 U + V C_1 - \lambda I & B_1 \\ C_1 & 0 \end{bmatrix} = \begin{bmatrix} A - \lambda I & B_1 \\ C_1 & 0 \end{bmatrix}
					\end{align*}
					Since multiplying by invertible matrices preserves rank, we have:
					\begin{equation}
						\operatorname{rank} \begin{bmatrix} A_0 - \lambda I & B_1 \\ C_1 & 0 \end{bmatrix} \!=\! \operatorname{rank} \begin{bmatrix} A - \lambda I & B_1 \\ C_1 & 0 \end{bmatrix}\ge \operatorname{rank}(A - \lambda I) \!=\! n.
					\end{equation}
					This proves the necessity of the condition.
					
					\vspace{0.5cm}
					
					\textbf{Sufficiency:}
Assume that for all \(|\lambda| \ge 1\), \eqref{weak-stability-noisy} holds. Let \(\Psi_1 \in \mathbb{R}^{n \times r}\) be a matrix whose columns form a basis for the left null space of \(B_1\) (meaning \(\Psi_1^T B_1 = 0\), and it has full column rank \(r\)). Let $\tilde B_1$ consist of $n-r$ columns of $B_1$ such that ${\rm im ( \tilde B_1)}={\rm im (B_1)}$.  Left-multiplying the block matrix by a non-singular matrix yields:
\begin{equation*}
\begin{bmatrix} \Psi_1^T & 0 \\ 0 & I_{n-r} \\ \tilde B_1^T & 0 \end{bmatrix} \begin{bmatrix} A_0 - \lambda I & \tilde B_1 \\ C_1 & 0 \end{bmatrix} = \begin{bmatrix} \Psi_1^T (A_0 - \lambda I) & 0 \\ C_1 & 0 \\ \tilde B_1^T (A_0 - \lambda I) & \tilde B_1^T \tilde B_1 \end{bmatrix}
\end{equation*}
As \(\tilde B_1^T \tilde B_1\) is invertible, the rank of this product matrix reduces to
\begin{align*}
&	n-r + \operatorname{rank} \begin{bmatrix} \Psi_1^T (A_0 - \lambda I) \\ C_1 \end{bmatrix} \ge n \\ &
\implies \operatorname{rank} \begin{bmatrix} \Psi_1^T A_0 - \lambda \Psi_1^T \\ C_1 \end{bmatrix} \ge  r, \quad \forall |\lambda| \ge 1.
\end{align*}

By the PBH test for detectability of rectangular descriptor systems \cite{darouach2002design,Hou1999Observer}, this rank condition guarantees the existence of a matrix \(V_2 \in \mathbb{R}^{r\times (n+m)}\) such that
\begin{equation}
\operatorname{rank} \left( \Psi_1^T(A_0 - \lambda I) + V_2 C_1 \right) =  r, \quad \forall |\lambda| \ge 1
\end{equation}

We can now define $V$ by back-projecting \(V_2\) as
\(V = \Psi_1(\Psi_1^T \Psi_1)^{-1} V_2\).
This assignment implies \(\Psi_1^T V = V_2\). Therefore, \(\Psi_1^T (A_0 + V C_1 - \lambda I)\) has full row rank for all \(|\lambda| \ge 1\). According to the PBH test, this guarantees that there is no unstable left eigenvector of \(A_0 + V C_1\) in the left null space of \(B_1\). Thus, the pair \((A_0 + V C_1, B_1)\) is Schur stabilizable. There exists a state feedback matrix U such that the closed-loop matrix \(A = (A_0 + V C_1) + B_1 U\)
is Schur stable. This completes the proof of sufficiency.
				\end{proof}

\begin{remark}
	In the noise-free case, i.e., $E = I_n$, we have $E_\perp = 0$ and $E^\dagger E = I_n$. Consequently, condition \eqref{weak-stability-noisy} reduces to condition \eqref{weak-stability}. Similarly, Corollary \ref{weak-controllability} reduces to the controllability condition in Theorem \ref{theo-weak}.
\end{remark}

The data $(U,X)$ are informative for controllability (resp. stability) on the noise set $\mathbb{D}$ if the system $(A,B)$ parameterized by \eqref{parameterize-noisy} is controllable (resp. stable) for \emph{all} matrices $U, V$ of compatible dimensions. Somewhat surprisingly, the following results demonstrate that conventional informativity can only be achieved for the linear subspace noise model $\mathbb D$ when $\mathbb{D} = \{0_{n\times T}\}$, i.e., when the data are noiseless. 

\begin{corollary}
	The data $(U,X)$ are informative for controllability on the noise set $\mathbb{D}$ if and only if $E_\perp = 0$ and $\operatorname{rank}(\lambda X_p - X_f) = n$ for all $\lambda \in \mathbb{C}$.
\end{corollary}

                \begin{proof}
                According to the PBH test, a system $(A, B)$ is controllable if and only if there does not exist any complex scalar $s \in \mathbb{C}$ and any non-zero vector $w \in \mathbb{C}^n$ (a left eigenvector of $A$) such that $w^T [A, B] = [s w^T, 0]$. Substituting the parameterization \eqref{parameterize-noisy} into this condition yields \begin{equation}
                    \label{eq:phb_noise}
                    w^T E_\perp U + w^T V Z_p^\perp = [s w^T, 0] - w^T E^\dagger E X_f Z_p^\dagger.
                \end{equation}
                 Suppose $E_\perp \neq 0$. Then there exists a non-zero vector $w$ such that $w \in \ker(E)$, meaning $w^T E_\perp = w^T \neq 0$.
Because $U$ is a completely free arbitrary matrix, $w^T E_\perp U = w^T U$ can span the entire space $\mathbb{C}^{1 \times (n+m)}$. By picking $V = 0$, we can easily select $U$ such that $w^T U$ exactly equals the right-hand side of \eqref{eq:phb_noise}. This would make the system uncontrollable for that specific $U$. To prevent this and guarantee controllability for all $U$, we must restrict the null space of $E$ to be empty. As $E$ has full row rank, $E$ is non-singular. Then, the informativity problem reduces to the conventional one without noise, and the required results follow from condition \eqref{rankcondition}. 
                \end{proof}

\begin{corollary}
	The data $(U,X)$ are informative for stability on the noise set $\mathbb{D}$ if and only if $E_\perp = 0$ and the conditions in Proposition \ref{stable-no-prior} hold.
\end{corollary}

\begin{proof}
	Following the logic in the proof of Proposition \ref{stable-no-prior}, for all realizations $A$ parameterized by \eqref{parameterize-noisy} to be stable regardless of the choices for $U$ and $V$, it is necessary that $E_{\perp} = 0$ (and consequently ${\mathbb D}=\{0\}$). Once $E_\perp = 0$, the required results follow precisely from Proposition \ref{stable-no-prior}.
\end{proof}

\begin{remark}
	If the available prior knowledge constitutes a single-affine constraint, one could follow a similar research line: first derive the parameterization for the intersection set $\hat{\Sigma}_{\mathrm{\rm stru}} \cap \Sigma^n_{{\mathcal D}}$, and then establish conditions under which all elements in this parameterized set are controllable or stable. However, parameterizing $\hat{\Sigma}_{\mathrm{\rm stru}} \cap \Sigma^n_{\mathcal D}$ entails solving heavily coupled matrix equations, which is significantly more intricate than the noise-free case. We leave such extensions for future work.
\end{remark}

				\section{Conclusions}

This paper developed a unified framework distinguishing between weak and strong data informativity under prior structural knowledge. Weak informativity describes the intermediate
regime between strong informativity for a property
and strong informativity for its complement. By explicitly parameterizing the set of systems jointly consistent with data and structural priors as affine functions of free parameters via the Moore--Penrose pseudoinverse, we characterized these notions for fundamental system properties, including controllability, stabilizability, and stability. For controllability, we also established a generic dichotomy: either no data-consistent system is controllable, or almost all such systems are controllable.  We further extended this framework to accommodate noisy data belonging to known linear subspaces.

           	Our work shows that, even when the data are not informative for property ${\mathcal P}$, they may still be weakly informative (or informative for the complementary property ${\mathcal P}^c$). Consistent with \cite{berberich2022combining,chiuso2025harnessing}, prior structural knowledge relaxes data excitation requirements, but also reduces the degree of freedom for system parameterization, making the analysis more complicated. 
            Future research directions will focus on controller design with prior structural knowledge and expanding these parameterization techniques to handle quadratic noise bound constraints.

 \section*{Appendix: Proofs of some technical results} 
			\textbf{Proof of Theorem \ref{theo-weak}:}
Let $A_0=X_fZ_{p_1}^\dagger$, $B_0=X_fZ_{p_2}^\dagger$, $C_1=Z_{p_1}^\perp$, and $D_1=Z_{p_2}^\perp$. Then $(A_Y,B_Y)$ can be written in the linearized parameterization as $A_Y=A_0+YC_1$ and $B_Y=B_0+YD_1$. Hence, the family of all data-consistent systems is a linearly parameterized family of the form
$
A_p=A_0+B_1P_1C_1,\, B_p=B_0+B_1P_1D_1,
$
with ${\bf k}=\{1\}$, $P_1=Y$, and $B_1=I_n$.

We now apply Lemma \ref{lemma:sc}. The graph $\mathcal{G}$ has only two nodes, namely nodes $0$ and $1$. The transfer matrix from node $0$ to node $1$ is $H_{10}(\lambda)=C_1(\lambda I-A_0)^{-1}B_0+D_1$. If $Z_p^\perp=0$, then $[C_1,D_1]=0$ and the family is a singleton and Condition (2) is its PBH controllability test. If $[C_1,D_1]\ne 0$, condition (1) of Lemma \ref{lemma:sc} reduces to $H_{10}(\lambda)\not\equiv 0$.

For condition (2), ${\bf k}=\{1\}$ gives only two subsets. When ${\bf s}=\emptyset$, the condition is
\[
\operatorname{rank}\begin{bmatrix}\lambda I-A_0&-B_0\\C_1&D_1\end{bmatrix}\ge n.
\]
When ${\bf s}=\{1\}$, the rank condition is automatic because the matrix contains $B_1=I_n$. Thus Lemma \ref{lemma:sc} reduces  Conditions (1)--(2).

				
				{\bf Proof of Corollary \ref{coro:weak-stabilizability}}: Let $A_0,B_0,C_1,D_1$ be defined in Remark \ref{remark_compare}.   Let $\mathbb{C}_u = \{\lambda \in \mathbb{C} : |\lambda| \ge 1\}$ denote the unstable region of the complex plane.  
				
			{\bf{Necessity:}} Stabilizability requires $\operatorname{rank}[\lambda I-A\ \ B]=n$ for all $\lambda\in\mathbb C_u$. Since $[\lambda I-A\ \ B]=[\lambda I-A_0\ \ B_0]-Y[C_1\ \ D_1]$,
				\[
				\operatorname{rank}[\lambda I-A\ \ B]\le\operatorname{rank}\begin{bmatrix}\lambda I-A_0&B_0\\C_1&D_1\end{bmatrix},
				\]
				which proves necessity of \eqref{eq:weak-stabilizability}.
				
				{\bf{Sufficiency:}} We establish sufficiency by checking the structural connection through the transfer matrix $H(\lambda) :=H_{10}(\lambda)= C_1(\lambda I - A_0)^{-1}B_0 + D_1$.
				
				{\bf{Case 1: $H(\lambda) \not\equiv 0$.}}
				In this case, the system is structurally path-connected from the inputs (c.f. Lemma \ref{lemma:sc}). From \cite{zhang2019structural}, the set of parameters $Y$ that result in uncontrollability at any non-fixed mode constitutes a proper algebraic variety. Since condition \eqref{eq:weak-stabilizability} holds, there are no fixed uncontrollable modes in $\mathbb{C}_u$. Thus, a generic choice of $Y$ ensures that any remaining uncontrollable modes lie 
				in the stable region $\mathbb{C} \setminus \mathbb{C}_u$.
				
				{\bf{Case 2: $H(\lambda) \equiv 0$.}}
				In this case,  $D_1 = 0$ and $C_1 A_0^j B_0 = 0$ for all $j \ge 0$. This implies the reachable subspace of $(A_0, B_0)$, denoted as $\mathcal{R}$, satisfies $\mathcal{R} \subseteq \ker(C_1)$. 
				
				Let $\operatorname{dim}(\mathcal{R}) = n_1$. We choose a Kalman canonical transformation matrix $T = [T_1 \,\, T_2]$ where the columns of $T_1$ span $\mathcal{R}$. Applying this coordinate transformation, the nominal system matrices decouple into:
				\begin{equation}
					\tilde{A}_0 = T^{-1}A_0T = \begin{bmatrix} A_{11} & A_{12} \\ 0 & A_{22} \end{bmatrix}, \tilde{B}_0 = T^{-1}B_0 = \begin{bmatrix} B_1 \\ 0 \end{bmatrix}
				\end{equation}
				where the sub-pair $(A_{11}, B_1)$ is reachable by construction. Since $T_1$ spans $\mathcal{R}$ and $\mathcal{R} \subseteq \ker(C_1)$, it follows that $C_1 T_1 = 0$. Therefore, $\tilde{C}_1 = C_1 T = \begin{bmatrix} 0 & C_{12} \end{bmatrix}$.
				
				Let us perform this same similarity transformation on the original rank condition \eqref{eq:weak-stabilizability}. We have
				\begin{align}
					&\operatorname{rank} \begin{bmatrix} \lambda I - A_0 & -B_0 \\ C_1 & 0 \end{bmatrix} \\
					&= \operatorname{rank} \begin{bmatrix} \lambda I - A_{11} & -A_{12} & -B_1 \\ 0 & \lambda I - A_{22} & 0 \\ 0 & C_{12} & 0 \end{bmatrix} \ge n_1 + n_2
				\end{align}
				Because $(A_{11}, B_1)$ is reachable, $\begin{bmatrix} \lambda I - A_{11} & -B_1 \end{bmatrix}$ exhibits full row rank $n_1$ for all $\lambda \in \mathbb{C}$. The above inequality reduces to
				\begin{equation} \label{eq:detectable}
					\operatorname{rank} \begin{bmatrix} \lambda I - A_{22} \\ C_{12} \end{bmatrix} \ge n_2, \quad \forall \lambda \in \mathbb{C}_u
				\end{equation}
				Equation \eqref{eq:detectable} corresponds to the PBH test for the detectability of the pair $(A_{22}, C_{12})$. 
				
				Now look at the parameterized state matrix $A = A_0 + Y C_1$. Transforming it to our canonical coordinates with a partitioned parameter matrix $T^{-1}Y = \begin{bmatrix} Y_1^T & Y_2^T \end{bmatrix}^T$ yields:
				\begin{equation}
					\tilde{A} = \tilde{A}_0 + (T^{-1}Y)\tilde{C}_1 = \begin{bmatrix} A_{11} & A_{12} + Y_1 C_{12} \\ 0 & A_{22} + Y_2 C_{12} \end{bmatrix}
				\end{equation}
				The uncontrollable modes of this system are the eigenvalues of the autonomous block $A_{22} + Y_2 C_{12}$. Since the pair $(A_{22}, C_{12})$ is detectable by virtue of \eqref{eq:detectable}, linear pole placement guarantees that there exists a specific choice of $Y_2$ that shifts all eigenvalues of $A_{22} + Y_2 C_{12}$  into the stable domain. Consequently, there exists a $Y$ for which $(A_Y,B_Y)$ is stabilizable, completing the proof of sufficiency.

					{\bf  Proof of Proposition \ref{auxiliary_lemma}: }
				Consider the linear operator $\mathcal{A}: \mathbb{R}^{n \times q} \to \mathbb{R}^r$ defined by $\mathcal{A}(W) = (\langle M_1, W \rangle, \dots, \langle M_r, W \rangle)^T$. The system can be written as $\mathcal{A}(W) = \gamma$.
				
			We first establish the consistency condition.  Let $\mathcal{M} = \operatorname{span}\{M_1, \dots, M_r\}$ and let $\mathcal{M}^\perp = \{X \in \mathbb{R}^{n \times q} : \langle M_i, X \rangle = 0, \forall i=1,\dots,r\}$. Then, any matrix $W \in \mathbb{R}^{n \times q}$ can be uniquely decomposed as $W = W_{\mathcal{M}} + W_{\mathcal{M}^\perp}$, where $W_{\mathcal{M}} \in \mathcal{M}$ and $W_{\mathcal{M}^\perp} \in \mathcal{M}^\perp$.
				Since $\langle M_i, W_{\mathcal{M}^\perp} \rangle = 0$, the vector $\gamma$ must lie in the image of the subspace $\mathcal{M}$ under $\mathcal{A}$. Let $W_{\mathcal{M}} = \sum_{j=1}^r \alpha_j M_j$ for some coefficients $\alpha \in \mathbb{R}^r$. Substituting this into \eqref{eq:system} yields:
				$ \sum_{j=1}^r \alpha_j \langle M_i, M_j \rangle = \gamma_i \implies G\alpha = \gamma. $
				This linear system for $\alpha$ has a solution if and only if $\gamma \in \operatorname{im}(G)$. 
				
	 Let $W_p = \sum_{j=1}^r \alpha_j M_j$ with $\alpha = G^\dagger \gamma$. Then
				$ \langle M_i, W_p \rangle = \sum_{j=1}^r (G^\dagger \gamma)_j \langle M_i, M_j \rangle = (G G^\dagger \gamma)_i. $
				If $\gamma \in \operatorname{im}(G)$, then $G G^\dagger \gamma = \gamma$, satisfying the non-homogeneous equations.
				 The homogeneous solution $W_h$ must lie in $\mathcal{M}^\perp$. We can characterize any $W_h \in \mathcal{M}^\perp$ as the orthogonal projection of an arbitrary matrix $Y \in \mathbb{R}^{n \times q}$ onto $\mathcal{M}^\perp$, given by $P_{\mathcal{M}^\perp}(Y) = Y - P_{\mathcal{M}}(Y)$. 
				
				To find $P_{\mathcal{M}}(Y)$, we note $P_{\mathcal{M}}(Y) = \sum_{i=1}^r c_i M_i$ for some $c$. The orthogonality condition $\langle M_k, Y - P_{\mathcal{M}}(Y) \rangle = 0$ implies $\sum_{i=1}^r c_i \langle M_k, M_i \rangle = \langle M_k, Y \rangle$, or $Gc = b$ where $b_j = \langle M_j, Y \rangle$. Solving for $c$ using the pseudoinverse gives $c_i = \sum_{j=1}^r (G^\dagger)_{ij} \langle M_j, Y \rangle$. Substituting this back:
				\[ P_{\mathcal{M}}(Y) = \sum_{i=1}^r \left( \sum_{j=1}^r (G^\dagger)_{ij} \langle M_j, Y \rangle \right) M_i. \]
				Thus, $W_h = Y - \sum_{i=1}^r \sum_{j=1}^r (G^\dagger)_{ij} \langle M_j, Y \rangle M_i$ represents all elements in the kernel of the system. Combining $W_p$ and $W_h$ yields \eqref{eq:solution}.

{\bf Proof of Theorem \ref{theo:informativity-controllability}:}
By Proposition~\ref{thm:block_structured_param}, every jointly consistent system has the form
\[
[A_Y,B_Y]=X_fZ_p^\dagger+(M_0+S_{\ell\perp}U+V\Sigma_2^\perp)Z_p^\perp,
\]
where $U,V\in\mathbb R^{n\times(n+m)}$ are arbitrary. The PBH test states that such a system is uncontrollable if and only if there exist $\lambda\in\mathbb C$ and $v\ne0$ such that
\begin{equation}\label{eq:pbh_start}
v^* \left( X_f Z_p^\dagger + (M_0 + S_{\ell \perp} U + V \Sigma_2^\perp) {Z_p^\perp} \right) = \lambda v^* E_n.
\end{equation}
Right-multiplying \eqref{eq:pbh_start} by $Z_p$ and noting ${Z_p^\perp} Z_p = 0$ and $X_fZ_p^\dagger Z_p=X_f$ yields $v\in\mathcal K_\lambda$. For such a vector, $v^*X_fZ_p^\dagger=\lambda v^*E_nZ_pZ_p^\dagger$, so~\eqref{eq:pbh_start} is equivalent to
\begin{equation}\label{eq:pbh_red}
v^*S_{\ell\perp}UZ_p^\perp+v^*VR
=v^*(\lambda E_n-M_0)Z_p^\perp.
\end{equation}
Column-wise vectorization converts this equation into $\mathcal L(v)\xi=b(\lambda,v)$, where $\xi=[\operatorname{vec}(U)^T,\operatorname{vec}(V)^T]^T$ must be real. Separating real and imaginary parts yields
\begin{equation}\label{eq:necessity_real_sys}
\mathcal L_{\mathbb R}(v)\xi=b_{\mathbb R}(\lambda,v).
\end{equation}
This real linear system is infeasible if and only if \eqref{eq:exact-strong-controllability} holds. Requiring infeasibility for every candidate $(\lambda,v)$ proves the equivalence.

\emph{Sufficiency of Conditions~1--2.}
Condition~1 implies $v^*S_{\ell\perp}=0$ for every $v\in\mathcal K_\lambda$. Since $RW=0$, multiplying~\eqref{eq:pbh_red} by $W$ would then give
$
v^*(\lambda E_n-M_0)Z_p^\perp W=0.
$
Together with $v\in\mathcal K_\lambda\setminus\{0\}$, this contradicts Condition~2. Thus, no jointly consistent system is uncontrollable.

\emph{Necessity when $S_\ell=I_n$.}
Here $S_{\ell\perp}=0$, so Condition~1 is automatic. If Condition~2 fails at some $\lambda\in\mathbb C$, there is a $v\in\mathcal K_\lambda\setminus\{0\}$ such that
\[
\begin{aligned}
z&:=v^*(\lambda E_n-M_0)Z_p^\perp,\qquad zW=0,\\
z&=zR^\dagger R,\qquad g:=zR^\dagger.
\end{aligned}
\]
It is therefore enough to find a real $V\in\mathbb R^{n\times(n+m)}$ satisfying $v^*V=g$, because then $v^*VR=gR=z$ and~\eqref{eq:pbh_red} holds with $U=0$. Write $v=v_{\rm R}+\mathrm{i}v_{\rm I}$ and $g=g_{\rm R}+\mathrm{i}g_{\rm I}$. If $v_{\rm R}$ and $v_{\rm I}$ are linearly independent, then
\[
\mathcal V:=\begin{bmatrix}v_{\rm R}^T\\v_{\rm I}^T\end{bmatrix},\qquad
G:=\begin{bmatrix}g_{\rm R}\\-g_{\rm I}\end{bmatrix},\qquad
V:=\mathcal V^\dagger G
\]
is real and satisfies $\mathcal V V=G$, since $\mathcal V$ has full row rank. Equivalently, $v_{\rm R}^TV=g_{\rm R}$ and $v_{\rm I}^TV=-g_{\rm I}$, so $v^*V=g$.

If $v_{\rm R}$ and $v_{\rm I}$ are linearly dependent, write $v=\alpha w$ with $\alpha\ne0$ and $w\in\mathbb R^n\setminus\{0\}$. Dividing the two equations defining the violating pair by $\bar\alpha$, we may replace $v$ by $w$. If $\lambda$ is real, then $z$ and $g$ are real, and $V=wg/(w^Tw)$ satisfies $w^TV=g$. If $\lambda$ is nonreal, separating real and imaginary parts gives
\[
\begin{aligned}
w^TX_p&=w^TX_f=0,\\
w^TE_nZ_p^\perp W&=w^TM_0Z_p^\perp W=0.
\end{aligned}
\]
Hence, for any real $\mu$, the pair $(\mu,w)$ also violates Condition~2: with $z_\mu:=w^T(\mu E_n-M_0)Z_p^\perp$ and $g_\mu:=z_\mu R^\dagger$, one has $z_\mu W=0$, and the real matrix $V=wg_\mu/(w^Tw)$ satisfies $w^TVR=g_\mu R=z_\mu$. Thus, in every case, a real $V$ yields a jointly consistent system that violates the PBH controllability test. Condition~2 is therefore necessary when $S_\ell=I_n$.

					\bibliographystyle{elsarticle-num}
					{\footnotesize
						\bibliography{yz}

@ARTICLE{6093751,
  author={Hou, Zhongsheng and Jin, Shangtai},
  journal={IEEE Transactions on Neural Networks}, 
  title={Data-Driven Model-Free Adaptive Control for a Class of MIMO Nonlinear Discrete-Time Systems}, 
  year={2011},
  volume={22},
  number={12},
  pages={2173-2188}}

@article{nemirovskii1993several,
  title={Several {NP}-hard problems arising in robust stability analysis},
  author={Nemirovskii, Arkadii},
  journal={Mathematics of Control, Signals and Systems},
  volume={6},
  number={2},
  pages={99--105},
  year={1993},
  publisher={Springer}
}

@article{vandenhof2026informativity,
  title   = {Data-informativity conditions for structured linear systems
             with implications for dynamic networks},
  author  = {Van den Hof, Paul M. J. and Shi, Shengling and
             Fonken, Stefanie J. M. and Ramaswamy, Karthik R. and
             Hjalmarsson, H{\aa}kan and Dankers, Arne G.},
  journal = {Automatica},
  volume  = {191},
  pages   = {113110},
  year    = {2026}
}

@unpublished{cao2026informativity,
  title  = {Data Informativity for Linear Systems with Partially Known
            Parameters},
  author = {Cao, Ao and Wang, Fuyong and You, Keyou},
  year   = {2026},
  note   = {Preprint}
}

@article{Kogan2025Robust,
  author    = {Kogan, Mark M. and Stepanov, Andrey V.},
  title     = {Robust {$H_\infty$} Control With Transients for {LTV} Systems Based on Prior Knowledge and Data},
  journal   = {Automatica},
  volume    = {173},
  pages     = {112083},
  year      = {2025}
}

@article{shakouri2026data,
  title={Data-driven stabilization using prior knowledge on stabilizability and controllability},
  author={Shakouri, Amir and van Waarde, Henk J and Baltussen, Tren MJT and Heemels, WPMH},
  journal={IEEE Transactions on Automatic Control},
  year={2026},
  publisher={IEEE}
}

@article{Breschi2023StructuralPriors,
  author    = {Breschi, Valentina and Sassella, Alessandro and Formentin, Simone},
  title     = {Data-Driven Design of Explicit Predictive Controllers With Structural Priors},
  journal   = {IEEE Control Systems Letters},
  volume    = {7},
  pages     = {1616--1621},
  year      = {2023}
}

@inproceedings{krishnan2021direct,
  title={On direct vs indirect data-driven predictive control},
  author={Krishnan, Vishaal and Pasqualetti, Fabio},
  booktitle={2021 60th IEEE Conference on Decision and Control (CDC)},
  pages={736--741},
  year={2021},
  organization={IEEE}
}

@article{eising2021informativity,
  title={Informativity of noisy data for structural properties of linear systems},
  author={Eising, J. and Trentelman, H. L.},
  journal={Systems \& Control Letters},
  volume={158},
  pages={105058},
  year={2021}
}

@article{huang2025datadriven,
  title={Data-driven stabilization of polynomial systems using density functions},
  author={Huang, Huayuan and Camlibel, M Kanat and Carloni, Raffaella and van Waarde, Henk J},
  journal={Automatica},
  volume={193},
  pages={113229},
  year={2026},
  publisher={Elsevier}
}

@article{zhu2024data,
  title={Data informativity for robust output regulation},
  author={Zhu, Lijun and Chen, Zhiyong},
  journal={IEEE Transactions on Automatic Control},
  volume={69},
  number={10},
  pages={7075--7080},
  year={2024},
  publisher={IEEE}
}

@article{van2022data,
  title={Data-driven dissipativity analysis: Application of the matrix S-lemma},
  author={Van Waarde, Henk J and Camlibel, M Kanat and Rapisarda, Paolo and Trentelman, Harry L},
  journal={IEEE Control Systems Magazine},
  volume={42},
  number={3},
  pages={140--149},
  year={2022},
  publisher={IEEE}
}

@article{darouach2002design,
  title     = {Design of Observers for Descriptor Systems},
  author    = {Darouach, Mohamed and Boutayeb, Mohamed},
  journal   = {IEEE Transactions on Automatic Control},
  volume    = {40},
  number    = {7},
  pages     = {1323--1327},
  year      = {1995},
  publisher = {IEEE}
}

@article{chiuso2025harnessing,
  title={Harnessing uncertainty for a separation principle in direct data-driven predictive control},
  author={Chiuso, A. and Fabris, M. and Breschi, V. and Formentin, S.},
  journal={Automatica},
  volume={173},
  pages={112070},
  year={2025}
}

@article{berberich2022combining,
  title={Combining prior knowledge and data for robust controller design},
  author={Berberich, Julian and Scherer, Carsten W and Allg{\"o}wer, Frank},
  journal={IEEE Transactions on Automatic Control},
  volume={68},
  number={8},
  pages={4618--4633},
  year={2022},
  publisher={IEEE}
}

@article{Hou1999Observer,
  author    = {Hou, M. and Muller, P. C.},
  title     = {Observer design for descriptor systems},
  journal   = {IEEE Transactions on Automatic Control},
  year      = {1999},
  volume    = {44},
  number    = {1},
  pages     = {164--169}
}

@inproceedings{miller2025data,
  title={Data-driven structured robust control of linear systems},
  author={Miller, Jared and Eising, Jaap and D{\"o}rfler, Florian and Smith, Roy S},
  booktitle={2025 European Control Conference (ECC)},
  pages={96--101},
  year={2025},
  organization={IEEE}
}

@article{zhang2023real,
  title={On real structured controllability/stabilizability/stability radius: Complexity and unified rank-relaxation based methods},
  author={Zhang, Yuan and Xia, Yuanqing and Zhan, Yufeng},
  journal={Systems \& Control Letters},
  volume={178},
  pages={105578},
  year={2023},
  publisher={Elsevier}
}

@ARTICLE{8358007,
  author={Liu, Chun and Jiang, Bin and Zhang, Ke},
  journal={IEEE Transactions on Systems, Man, and Cybernetics: Systems}, 
  title={Adaptive Fault-Tolerant H-Infinity Output Feedback Control for Lead–Wing Close Formation Flight}, 
  year={2020},
  volume={50},
  number={8},
  pages={2804-2814}}

@article{kang2023minimum,
  title={Minimum input design for direct data-driven property identification of unknown linear systems},
  author={Kang, Shubo and You, Keyou},
  journal={Automatica},
  volume={156},
  pages={111130},
  year={2023},
  publisher={Elsevier}
}

@ARTICLE{informativity2020,
  author={van Waarde, Henk J. and Eising, Jaap and Trentelman, Harry L. and Camlibel, M. Kanat},
  journal={IEEE Transactions on Automatic Control}, 
  title={Data Informativity: A New Perspective on Data-Driven Analysis and Control}, 
  year={2020},
  volume={65},
  number={11},
  pages={4753-4768}}

@ARTICLE{1101135,
  author={Corfmat, J. and Morse, A.},
  journal={IEEE Transactions on Automatic Control}, 
  title={Structurally controllable and structurally canonical systems}, 
  year={1976},
  volume={21},
  number={1},
  pages={129-131}}

@article{2013Model,
  title={Model free adaptive control: Theory and applications},
  author={ Hou, Z.  and  Jin, S. },
  year={2013},
}

@ARTICLE{formulas2020,
  author={De Persis, Claudio and Tesi, Pietro},
  journal={IEEE Transactions on Automatic Control}, 
  title={Formulas for Data-Driven Control: Stabilization, Optimality, and Robustness}, 
  year={2020},
  volume={65},
  number={3},
  pages={909-924}}

@misc{lofberg2026static,
  title         = {Static Output-Feedback Stabilization Is {NP}-Hard},
  author        = {L{\"o}fberg, Johan},
  year          = {2026},
  eprint        = {2609.16886},
  archivePrefix = {arXiv},
  primaryClass  = {math.OC},
  url           = {https://arxiv.org/abs/2609.16886}
}

@misc{ahmadi2026complexity,
  title         = {Complexity of Output Feedback Stabilization},
  author        = {Ahmadi, Amir Ali and Chaudhry, Abraar and Narang, Ijay and Tang, Yukai},
  year          = {2026},
  eprint        = {2609.20636},
  archivePrefix = {arXiv},
  primaryClass  = {math.OC},
  url           = {https://arxiv.org/abs/2609.20636}
}

@book{Modern_Control_Ogata,
title={Modern Control Engineering},
  author={Ogata, Katsuhiko and Yang, Yanjuan},
  volume={4},
  year={2002},
  publisher={Prentice-Hall}
}

@article{zhang2019structural,
  title={Structural Controllability of an {NDS} with {LFT} Parameterized Subsystems},
  author={Zhang, Yuan and Zhou, Tong},
  journal={IEEE Transactions on Automatic Control},
  year={2019},
  volume={64},
  number={12},
  pages={4920--4935},
  publisher={IEEE}
}

@article{anderson1982structural,
  title={Structural controllability and matrix nets},
  author={Anderson, Brian DO and Hong, Hui-min},
  journal={International Journal of Control},
  volume={35},
  number={3},
  pages={397--416},
  year={1982},
  publisher={Taylor \& Francis}
}

@article{zhang2024zero,
  title={Zero-norm distance to controllability of linear dynamic networks},
  author={Zhang, Yuan and Xia, Yuanqing and Zhan, Yufeng and Sun, Zhongqi},
  journal={IEEE Transactions on Cybernetics},
  volume={54},
  number={12},
  pages={7368--7380},
  year={2024},
  publisher={IEEE}
}
			}   }}
		\end{document}